%% file: main.tex
\documentclass{article}
\usepackage[utf8]{inputenc}
\usepackage[english]{babel}
\usepackage{booktabs} 
\usepackage[ruled]{algorithm2e} 

\SetAlFnt{\small}
\SetAlCapFnt{\small}
\SetAlCapNameFnt{\small}
\SetAlCapHSkip{0pt}
\IncMargin{-\parindent}
\usepackage{hyperref}
\usepackage{url}
\usepackage{color-edits}
\usepackage{graphicx}
\usepackage{amsmath,amssymb}
\usepackage{amsthm}
\usepackage{fullpage}
\usepackage{bm}
\usepackage{mathrsfs}
\usepackage{appendix}
\usepackage{graphicx}
\usepackage{pifont}
\usepackage{xcolor}
\usepackage{xcolor}
\usepackage{algpseudocode}
\usepackage{bbm}
\usepackage{comment}
\usepackage{authblk}

\newtheorem{thm}{Theorem}
\newtheorem{lemma}{Lemma}
\newtheorem{definition}{Definition}
\newtheorem{cor}{Corollary}
\newtheorem{rmk}{Remark}
\newtheorem{fact}{Fact}
\newtheorem{clm}{Claim}

\DeclareMathOperator*{\argmax}{argmax}

\DeclareMathOperator{\E}{\mathbb{E}}

\renewcommand{\epsilon}{\varepsilon}

\newcommand{\D}{\mathcal{D}}

\newcommand{\bigo}{\mathcal{O}}
\newcommand{\opt}{\text{OPT}}
\newcommand{\A}{\mathcal{A}}
\newcommand{\M}{\mathcal{M}}

\newcommand{\R}{\mathcal{R}}
\newcommand{\Real}{\mathbb{R}}
\newcommand{\nb}{n^b}
\newcommand{\ns}{n^s}
\newcommand{\vs}{\mathbf{v}^s}
\newcommand{\vb}{\mathbf{v}^b}
\newcommand{\us}{\mathbf{u}^s}
\newcommand{\ub}{\mathbf{u}^b}
\newcommand{\oo}{\mathbf{a}}
\newcommand{\os}{\mathbf{a}^s}
\newcommand{\ob}{\mathbf{a}^b}
\newcommand{\ts}{\mathbf{r}^s}
\newcommand{\tb}{\mathbf{r}^b}

\newcommand{\wb}{\mathbf{w}^b}

\renewcommand{\S}{\mathcal{S}}
\newcommand{\B}{\mathcal{B}}
\newcommand{\1}{\mathbf{1}}

\newcommand{\umodb}{\mu^b}

\title{Differentially Private Call Auctions and Market Impact}
\author{Emily Diana, Hadi Elzayn, Michael Kearns, Aaron Roth, \\ Saeed Sharifi-Malvajerdi, and Juba Ziani}
\affil{University of Pennsylvania}

\date{\today}

\begin{document}

\maketitle

\begin{abstract}
\input{abstract}
\end{abstract}

\input{intro}
\input{prelims}
\input{our_mechanism}
\input{our_mechanism_lottery}
\input{lowerbound}
\input{connections_to_finance_lit}
\input{GAME-SEC}

\input{noregret_dynamics}

\section{Simulations} \label{sec:sim}

In previous sections, we designed our mechanism and obtained theoretical guarantees of performance; these guarantees were given both in a \emph{one-shot} setting and relative to the optimal result that could be reached given agents' bids and also in a \emph{repeated} setting using no-regret learning. In this section, we conduct experiments on simulated data in both a one-shot and learning setting in order to explore how tightly these guarantees bind in practice. 

We perform all simulations in MATLAB using a similar starting configuration. We have 5000 buyers and 5000 sellers, and valuations must be integer values between 1 and 100. Valuations are drawn from normal distributions centered at 45 for sellers and 55 for buyers, with standard deviations of 15 for both. The draws are rounded to the nearest integer, and draws below 1 or above 100 are replaced with 1 and 100 respectively.

The mechanism we implement is the first we defined (Algorithm \ref{alg:mech}), run once or repeatedly for the one-shot game and learning settings, respectively. We vary $\epsilon$ over a range from $\epsilon=0.01$ to $\epsilon=0.5$.

\begin{figure}[t]
\centering
\includegraphics[width=0.49\textwidth]{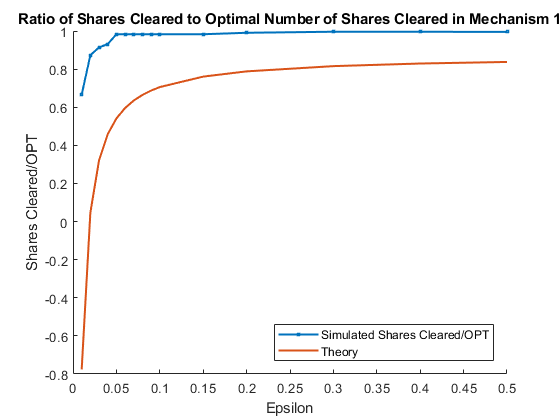}
\includegraphics[width=0.49\textwidth]{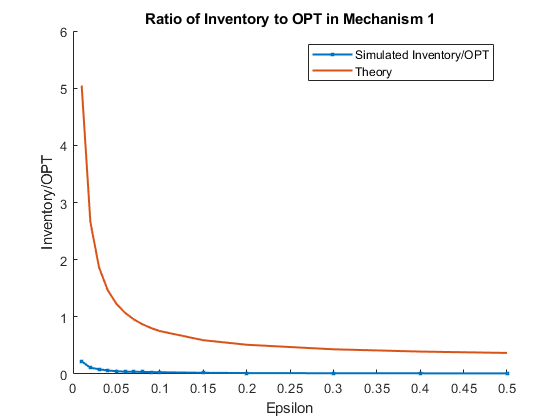}
\caption{Realized payoff and inventory relative to theoretical optimal in one-shot game for varying $\epsilon$.}\label{fig:oneshot}
\end{figure}

\paragraph{Single-shot game}
For the single shot game, we perform 800 trials per value of $\epsilon$ with a fixed set of agent valuations. These valuations were drawn randomly according to the procedure described above. We assume agents bid truthfully (and all of our comparisons are to the truthful optimal). 

In the first plot of Figure \ref{fig:oneshot}, we show the empirical 5\% quantile (i.e. the value for which only 5\% of draws saw lower values) of the \emph{competitive ratio} defined as the shares cleared as a fraction of $\opt$, the optimal number of shares that can be cleared given the realized valuations. This competitive ratio quantile is plotted in blue. The appropriate guarantee to compare to is the lower bound on this quantile given in Theorem~\ref{thm:thm_1}, with confidence parameter $\alpha = 0.05 / 8$; this bound is plotted in orange. While the realized competitive ratio indicates, unsurprisingly, that privacy is not costless for very small levels of $\epsilon$, it remains far above the worst-case guarantee predicted, and rapidly increases to nearly $1$ in the practical regime (i.e.,  even for $\epsilon=0.1$). This shows that, for a large enough number of agents and valuations drawn from well-behaved distributions, reasonable privacy can be achieved in practice with very little loss in utility.  

The second plot of Figure \ref{fig:oneshot} shows the inventory taken on by the mechanism, again plotting this quantile as a ratio of the optimal number of shares cleared in blue (again, limited to the top 95\% of runs) and the theoretical upper bound for $\alpha=0.05/6$ (as per the inventory bound of Theorem~\ref{thm:thm_1}) in orange. Notice that for very small $\epsilon$, the theoretical guarantee can be extremely large; yet, again, the realized inventory is far below the guarantee and never exceeds 23\% for even $\epsilon=0.01$ and is less than 5\% for $\epsilon \geq 0.05$.

\begin{figure}[t]
\centering
\includegraphics[width=0.49\textwidth]{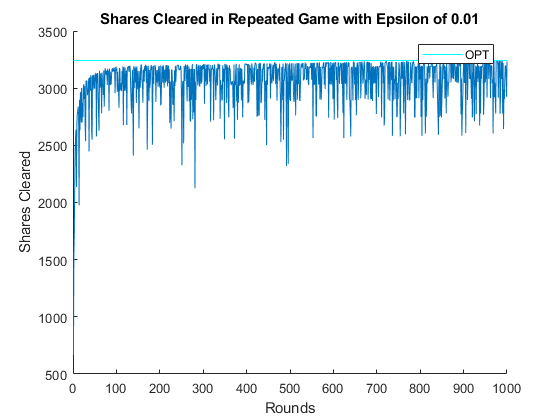}
\includegraphics[width=0.49\textwidth]{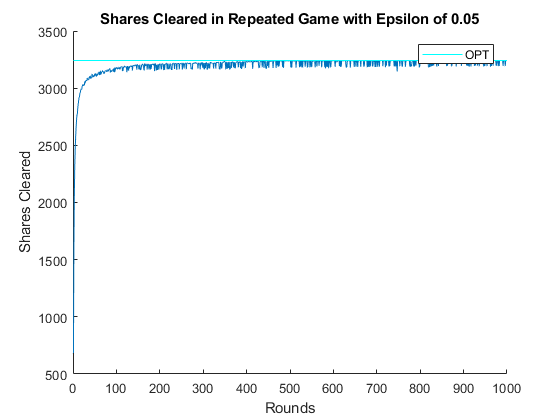}
\includegraphics[width=0.49\textwidth]{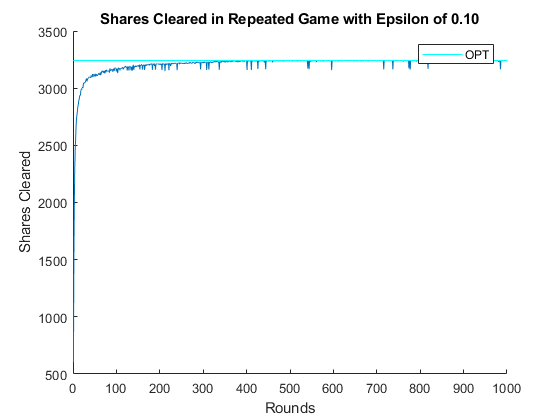}
\includegraphics[width=0.49\textwidth]{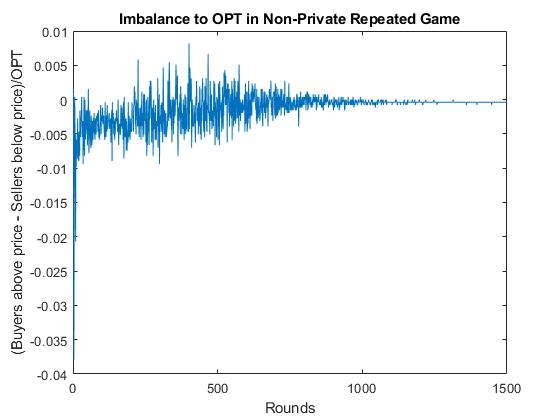}
\caption{The first three plots show the shares cleared over time, in the repeated setting of our mechanism, using Social Exponential Weights (\ref{alg:mw_variant}) for various choices of $\epsilon$. The last plot shows the imbalance between buyers and sellers over time in a repeated (non-private) auction. The agents' updates use $\eta, \xi = 0.1$.}\label{fig:repeated}
\end{figure}

\paragraph{Learning Setting}
In the \emph{learning} setting, we plot the shares cleared \emph{over time} as agents learn to bid given their valuations. We repeat the auction for 1000 rounds (1500 for the imbalance plot) with learning rate $\eta = 0.1$ and ``fake'' utility $\xi = 0.1$. Agents draw fixed valuations and then use the Social Exponential Weights described in Algorithm \ref{alg:mw_variant} to learn and bid each round. We repeat this process for several different values of $\epsilon$.

The first three plots in Figure \ref{fig:repeated} tell similar stories: agents, and thus the system, learn to bid over time in such a way as to clear the optimal number of shares (were the mechanism privacy-free). The noisiness in the plots is due to privacy and depends on the choice of $\epsilon$: the smaller the value of $\epsilon$, the more likely the mechanism is to pick a sub-optimal price, even after agents learn to bid optimally. For $\epsilon=0.01$, the randomness of the mechanism induces enough noise as to occasionally forego a large portion of utility; at larger values of $\epsilon$, the added randomness costs relatively little. 

The fourth plot displays the imbalance between number of buyers bidding above vs. sellers bidding below the price chosen by the repeated \emph{standard} (i.e., non-private) call auction when buyers use Social Exponential Weights. We highlight an interesting connection to real-world behavior: NYSE and NASDAQ perform  pre-opening or pre-closing repeated "hypothetical" auctions aimed at price discovery. In these hypothetical auctions, the exchanges accept bids, announce the current price and imbalance, allow bidders to submit updated bids, and repeat. The pattern in imbalances documented by \cite{challet2018dynamical} agrees broadly with that of Figure \ref{fig:repeated}: that is, the imbalance begins skewed to one side or another, but it repeatedly oscillates as bidders adjust before converging to a settled state.

\bibliographystyle{apalike}
\bibliography{ref}

\appendix
\input{appendix}
\input{appendix_dynamics}

\end{document}

%% file: abstract.tex
We propose and analyze differentially private (DP) mechanisms
for call auctions as an alternative to the complex and ad-hoc privacy efforts that are
common in modern electronic markets. We prove that the number of shares cleared in the
DP mechanisms compares favorably to the non-private optimal and provide a matching
lower bound. We analyze the incentive properties of our mechanisms and their 
behavior under natural no-regret learning dynamics by market participants. We include simulation results and connections to the finance literature on market impact.

%% file: intro.tex
\section{Introduction and Overview of Paper}

In modern financial markets, massive resources are directed
towards what can be considered ad-hoc privacy mechanisms, intended to allow participants to
cloak their trading activity and intentions. Such efforts occur both in the exchanges
themselves and in the algorithmic trading services offered by large brokerages.
In this work, we provide a differentially private (DP) version of classical one-shot
double auctions (also known as ``call auctions''). Frequent
instances of DP call auctions could potentially simplify the convoluted efforts
at providing trading secrecy that are rampant in today's markets while still permitting
dynamic price discovery.

Current electronic exchanges offer a staggering variety of order
types and mechanisms meant to provide specific types of privacy. Dark pools were
introduced to allow large-volume counterparties to discover each other away from the
so-called ``lit'' markets where high-frequency traders (HFTs) are prevalent. Order types
restricting execution with small-volume counterparties are meant to provide similar protections.
Hidden and ``iceberg'' orders in the lit exchanges provide secrecy at the expense
of time priority in the standard continuous limit order book. The relatively new exchange IEX
was created to foil the latency arbitrage of HFT by introducing a ``speed bump''
for all incoming orders. On the brokerage side, algorithms executing large client trades
attempt to minimize visibility by breaking orders up over time and across exchanges and
employ randomization in both timing and sizing to avoid detectable ``heartbeats.''

These efforts are all ad-hoc in the sense that they each protect market participants
from rather specific forms of detection or exploitation. While well-intentioned, they have contributed significantly to the
complexity of modern electronic markets. At the same time, it is also widely
understood that there are limits to the privacy that can be provided for large trades
executed in short periods, and there is a large academic and practical literature on
theories of market impact (see~\cite{gatheral2010no,gatheralslides} for an overview) 
and algorithms for minimizing it. This literature identifies
a trade's {\em participation rate\/} --- the ratio of its volume to that of the overall
market during the trade's execution --- as the key determinant of market impact.

Our main conceptual contribution is the development of DP call auctions as a 
mechanism providing privacy against {\em all\/} forms of attack or detection, up to the
participation rate of a trade. In this formulation, we provide a per-share privacy guarantee
determined by the sensitivity of the call auction, which in turn
determines the amount of noise added. Trades with higher participation rates will 
unavoidably have less privacy than those with smaller ones, but the nature of the privacy
will now be as general as possible. Repeated DP call auctions also enjoy graceful degradation
of the privacy guarantee.
Furthermore, we can (informally) relate our results to standard market
impact theories via the shared notion of participation rate and show that, under natural
conditions, DP call auctions clear a near-optimal number of shares under the predictions
of the ``square root law'' of market impact.
We analyze our DP mechanisms extensively, including its incentive properties and behavior
under natural no-regret learning dynamics by market participants.

We note that (non-private) call auctions are already common in modern markets. In particular,
both NYSE and NASDAQ hold call auctions (also sometimes called ``crosses'') to establish
opening and closing prices in U.S. equities~\cite{nyse,nasdaq};
in the Tokyo Stock Exchange there are additional intraday call auctions, which are also the subject of academic study (e.g.~\cite{challet2018dynamical,challet2019strategic}).
The influential paper~\cite{budish} (discussed at greater length in Related Work below)
proposes and analyzes frequent intraday (again non-private) call auctions specifically as a defense
against latency arbitrage; see also~\cite{WahWellman13}. Our work can be seen as a continuation of this line of thinking, in
which frequent intraday DP call auctions could provide even more general privacy guarantees to all
market participants.

\vspace{3mm}
\noindent\textbf{Outline and Summary of Results:} At a high level, our results fall into three broad categories:
\begin{enumerate}
    \item \textbf{The development and analysis of (jointly) differentially private call auctions}. We carry this out in Section \ref{sec:mechanisms}. We initially present this purely as an algorithm design task, abstracting away incentive properties. We prove bounds relating the \emph{privacy} properties of the mechanism, the \emph{number} of shares it is guaranteed to clear compared to the optimal benchmark, and the net \emph{inventory} that the mechanisms may have to take on. (Unavoidably, jointly differentially private call auctions cannot exactly match the number of buyers and sellers and so will have to take on a net position of shares itself to clear the market --- we prove that this net position is small.). We also prove a lower bound showing that our mechanisms are near optimal amongst all differentially private mechanisms. We explore the connection between our guarantees and theories of market impact in Section \ref{sec:marketimpact}.
    \item \textbf{The analysis of incentive properties and learning dynamics}. Having developed our algorithms, we turn our attention to how buyers and sellers should interact with them. First, in Section \ref{sec:game}, we show that our algorithm is ex-post individually rational and approximately dominant strategy truthful for agents who wish to trade only a small number of shares, with a guarantee that degrades gracefully in the size of the desired trade. (We note that this is a \emph{stronger} incentive guarantee than standard non-private call auctions.) We then study the global behavior that results when agents interact with a repeated version of one of our mechanisms using \emph{learning dynamics}: we show that although an abstract guarantee of no-regret learning is not enough to guarantee convergence to the optimal number of trades, a small modification of the exponential weights learning algorithm (informally, a modification that still guarantees the no-regret property, but breaks ties in favor of trading whenever such ties exist) does converge to the optimal number of trades. 
    \item \textbf{Simulation Results}. Finally, in Section \ref{sec:sim}, we conduct simulations in both \emph{one-shot} and \emph{repeated} settings, showing that in the settings considered, the realized outcomes of our mechanisms tend to be significantly better than the worst-case guarantees of our theorems. 
\end{enumerate}

\vspace{3mm}
\noindent\textbf{Related Work:} Our work relates to several large strands of literature. Prominently, the study of double auctions dates back to the early days of mathematical economics. \cite{parsons2006everything} provides an introduction to double auctions, and a useful survey from a computer science perspective can be found in \cite{parsons2011auctions}. Our modeling of the strategic framework in which agents participate in the double auction is broadly consistent with this literature. 

Of particular note are~\cite{budish} and \cite{WahWellman13}, which both propose frequent call auctions to eliminate latency arbitrage.\footnote{ \emph{Latency arbitrage} is the opportunity for traders to simultaneously buy and sell nearly or exactly identical securities on different exchanges (e.g. Chicago's Mercantile Exchange and the New York Stock Exchange) in the instant where price has changed on one exchange but remains ``stale'' on the other; it is described in the popular book \emph{Flash Boys} \cite{lewis2014flash}.} 
The work of~\cite{budish} first establishes the \emph{empirical} availability of latency arbitrage opportunities for even highly traded securities, and shows moreover that competition between traders has not eliminated this opportunity over time. Instead, it has resulted in an ``arms race'' for speed, with arbitrage windows becoming shorter over time, but arbitrage profit per unit  remaining essentially constant. 
The authors then propose a solution to mitigate latency arbitrage: repeated high-frequency call auctions.  Using a game theoretic approach, they model how the ``sniping'' process results in arbitrage opportunities in the continuous limit order book; in their model, the profit opportunity (along with arms race) is an equilibrium constant, even despite improving technology. Then, using the same underlying model of firm behavior, the authors show that repeated call auctions eliminate these arbitrage opportunities and cause firms not to choose to invest in speed, ending the arms race. We follow in the spirit of \cite{budish}, but note that their solution does not mitigate the problem of privacy, and in particular does not solve the issue of the proliferation of ad-hoc and increasingly complex trading algorithms. (The earlier work of~\cite{WahWellman13} performs extensive simulation studies that establish the salutary effects of frequent call auctions on latency arbitrage.)

Our work is connected to, and leverages tools from, the broad literature on differential privacy introduced by \cite{dwork}; for an overview, see, e.g. \cite{aaron}. The most related strand of this literature is the connection between differential privacy and mechanism design, first made by \cite{McSherryT07}. In particular, they observed that differentially private mechanisms inherit strong incentive properties. For many mechanism design tasks that involve the allocation of a resource to individuals, it is not possible to satisfy differential privacy in the standard sense over allocations: in cases like this, the relevant solution concept is \emph{joint differential privacy} \cite{jointdp}. This solution concept has been used in a number of mechanism design settings, including max-welfare matchings and other allocation problems \cite{billboard,hsu2016jointly}, stable matchings \cite{stableDP}, equilibrium selection problems \cite{RR14,cummings2015privacy}, and tolling problems \cite{rogers2015inducing}. In particular, although joint differential privacy can be used as a tool to achieve truthfulness, not all jointly private mechanisms are approximately truthful, and more specialized arguments are needed. Finally, while ~\cite{chen} have shown how to privately compute near-optimal prices in double auctions, their process does not guarantee end-to-end joint differentially privacy when taking trade allocations into account, unlike this work.

%% file: prelims.tex
\section{Model and Preliminaries}\label{sec:prelim}

\subsection{Model}

We consider a call auction setting with $\ns$ sellers and $\nb$ buyers; we let $\S$ be the set of sellers, $\B$ be the set of buyers, and $n = \ns + \nb$. Each seller $i \in \S$ has one unit of a security for which it has a \emph{value} $v_i^s$; each buyer $j \in \B$ wishes to purchase one unit of the security for which it has a value of $v_j^b$. We let $\vs = (\vs_1,\ldots,\vs_{\ns} )$ be the vector of all sellers' valuations and $\vb = (\vb_1,\ldots,\vb_{\nb} )$ the vector of all buyers' valuations. We assume valuations are drawn from a discrete set $P$; without loss of generality, we let $P=\{1,2,\ldots,V\}$ for some integer $V$.

Agents report their valuations in $P$ directly to a mechanism $\M$.\footnote{We will argue in Section~\ref{sec:game} that it is in every agent's best interest to report his valuation to the mechanism truthfully, hence our mechanisms can work with the agents' valuations without loss of generality.} Based on the agents' reports, the mechanism selects a clearing price $p \in P$ and an allocation vector $\oo = (\os, \ob)$, where $\os_i$ (resp. $\ob_j$) is equal to $1$ if seller $i$ (resp. buyer $j$) is selected to participate in a trade and $0$ otherwise. The mechanism concludes by buying a share at price $p$ from every seller $i$ with $\os_i = 1$ and selling a share at price $p$ to every buyer with $\ob_j = 1$.\footnote{In principle, mechanisms can choose \emph{non-uniform} pricing; that is, different agents could be charged different prices based on their reports. Here, we only consider \emph{uniform} pricing mechanisms, as is commonplace in the double auction literature.}

\paragraph{Privacy Constraints} 
The outcomes of the mechanisms we consider are functions of the agents' reports, which themselves depend on their valuations. In turn, these outcomes may leak information about the participants' valuations. This provides motivation for designing call auctions that protect the privacy of the participants. In this paper, we do so using  \emph{differential privacy} (\cite{dwork}). We will design our mechanisms to release the clearing price $p$ in a differentially private fashion, and the allocation vector $\oo = (\os, \ob)$ in a \emph{jointly} differentially private manner \cite{jointdp}. Differential privacy and joint differential privacy are formally defined in Section~\ref{sec: dp}.

\paragraph{Mechanism Designer's Objective} The main objective of our mechanisms for call auctions is to maximize the volume of trades between buyers and sellers. However, because of the randomization that we will need to add to achieve  differential privacy, our mechanism will inevitably incur several kinds of cost. First, the \emph{payoff} of the mechanism, given by the number of shares cleared, will generally be lower than the optimal payoff that could have been reached absent differential privacy. Second, we will have to deal with situations in which the number of sellers and the number of buyers who are selected to trade differ because of the noise added to the allocation rule for privacy concerns; this creates an \emph{inventory} in which some of the trades must be fulfilled by the mechanism itself (when there are more sellers selected than buyers, the mechanism buys surplus shares from the sellers; when there are more buyers selected than sellers, the mechanism sells to the buyers from its own reserve of shares). We will aim to keep the inventory of our private auction mechanisms as small as possible. Formally, the payoff and the inventory of a mechanism $\M$ are defined as follows:
\begin{definition}[Payoff and Inventory of a Mechanism $\M$]\label{def:utility_inventory}
For any mechanism $\M$ outputting a price $p$ and an allocation vector $\oo$, the payoff is the number of shares cleared by $\M$:
$$
\Pi \left( \M \right) = \min \left\{ \sum_{i \in \S} \os_i , \sum_{j \in \B} \ob_j \right\}
$$,
and the inventory of $\M$ is the number of allocations that must be fulfilled by the mechanism:
$$
I \left( \M \right) = \left\vert \sum_{i \in \S} \os_i - \sum_{j \in \B} \ob_j \right\vert
$$.
\end{definition}

 The main benchmark we use to measure the performance of our mechanisms is the maximum number of trades that can be obtained (absent differential privacy) while setting a uniform price $p$ and guaranteeing every agent non-negative utility.\footnote{i.e., we only allocate agents willing to trade at price $p$. Agents not willing to participate at price $p$ will opt out from the trade, thus are not taken into account by our benchmark.} Formally, our benchmark is given by\footnote{$\1 [A]$, here and throughout the paper, represents the indicator function of event $A$.}
\begin{equation}\label{eq:opt}
 \opt = \max_{p \in P} \min \left\{\sum_{i \in \S} \1 \left[\vs_i \leq p\right], \sum_{j \in \B} \1 \left[\vb_j \geq p\right] \right\}.
\end{equation}

\subsection{Differential Privacy}\label{sec: dp}

Let $\D$ be a \emph{data universe} from which a data set $D$ of size $n$ is drawn. In the setting considered in this paper, $D= (\vs,\vb)$ contains the reported valuations of sellers and buyers in the market. The algorithms we consider in this paper have output that can naturally be partitioned across the $n$ users who provide the inputs --- namely for each agent, whether they get to participate in a trade, and at what price. Let $\M$ be an algorithm that takes the data set $D$ as input and outputs $\M(D) \in \R^n$, which is a vector whose $i$th coordinate corresponds to the output sent to agent $i$. Here $\R$ is the output range of the algorithm for a single agent, which we will take to  be $\{0,1\} \times P$ (whether someone is chosen to participate in a trade, and a price for the trade). Informally speaking, differential privacy requires that a change in a single data entry should have little (distributional) effect on the output of the mechanism. In other words, for every pair of data sets $D,D' \in \D^n$ that differ in at most one entry, differential privacy requires that the distribution of $\M(D)$ and $\M(D')$ are ``close" to each other where closeness is measured by the privacy parameters $\epsilon$ and $\delta$.
\begin{definition}
Let $D, D' \in \D^n$ be two data sets of size $n$. We say $D$ and $D'$ are neighboring and write $D \sim D'$ if they differ in at most one data entry. $D$ and $D'$ are called $i$-neighbors ($D \sim_i D'$) if $D_{-i} = D'_{-i}$.
\end{definition}

\begin{definition}[(Standard) Differential Privacy (DP) \cite{dwork}]
An algorithm $\M: \D^n \to \R^{n}$ is $(\epsilon,\delta)$-differentially private if for every pair of neighboring data sets $D \sim D' \in \D^n$, and for every subset of outputs $S \subseteq \R^{n}$,
	\begin{align*}
	\Pr \left[\mathcal{M}(D) \in S \right] \leq e^\epsilon \cdot \Pr \left[\mathcal{M}(D') \in S \right] + \delta
	\end{align*}
where the probability is taken with respect to the randomness of $\M$. if $\delta = 0$, $\M$ is said to be $\epsilon$-DP.
\end{definition}

We now define \emph{joint differential privacy}. Joint differential privacy is defined in settings in which not only the inputs but also the outputs of the mechanism can be partitioned amongst the $n$ users of the mechanism. In our setting, as in many mechanism design settings, this is the case: users report their valuations (which constitute the data) and then each receives an individual allocation. Joint differential privacy requires that an individual's input to the mechanism has little (distributional) effect on the outputs given to \emph{others} --- but allows one's own input to have a large effect on one's own output. Informally, it protects the privacy of each individual from arbitrary coalitions of other individuals using the system. 
\begin{definition}[Joint Differential Privacy \cite{jointdp}] An algorithm $\mathcal{M}: \D^n \to \R^{n}$ is $(\epsilon,\delta)$-joint differentially private if for every $i$, for every pair of $i$-neighbors $D \sim_i D' \in \D^n$, and for every $S \subseteq \R^{n-1}$,
	\begin{align*}
	\Pr \left[\mathcal{M}(D)_{-i} \in S \right] \leq e^\epsilon \cdot  \Pr \left[\mathcal{M}(D')_{-i} \in S \right] + \delta
	\end{align*}
where the probability is taken with respect to $\M$'s randomness. If $\delta = 0$, $\M$ is said to be $\epsilon$-joint DP.
\end{definition} 

We will use the \emph{Laplace} and \emph{exponential mechanisms} of differential privacy in our proposed algorithms. See Appendix~\ref{app:dp_tools} for their formal definitions, their privacy and accuracy guarantees, and a few properties of differential privacy including \emph{post-processing} and \emph{composition}.

%% file: our_mechanism.tex
\section{Private Call Auction Mechanisms}\label{sec:mechanisms}

In this section, we outline our jointly differentially private mechanisms for the call auction problem and analyze their performance guarantees. Each mechanism's performance is measured in terms of its \emph{payoff} --- that is, the total number of shares cleared --- as well as its \emph{inventory} --- the net position that the mechanism must itself take on. We measure our mechanisms' payoffs against the \emph{maximal} number of shares that could be cleared with a uniform price, \emph{given} the agents' reports. 

Throughout this section, we assume reports are truthful; we will show in Section \ref{sec:game} that
our mechanisms are approximately dominant strategy truthful.
We also highlight that taking on \emph{some} inventory is unavoidable -- if the mechanism took no net position, a coalition of agents could use the constraint that the number of buyers and sellers must be equal to circumvent joint differential privacy -- but our guarantees ensure that this net position remains small with high probability. 

We propose three mechanisms. The first mechanism, described in Subsection \ref{subsec:coin}, uses the exponential mechanism (see Appendix \ref{app:dp_tools}) to select a clearing price and then uses binomial randomization to determine who participates in a trade. In Subsection \ref{subsec:lottery}, we provide a second mechanism that again uses the exponential mechanism to select a price, but uses lottery numbers which are assigned to agents ex-ante to determine market participants. In Subsection \ref{subsec:meta}, we describe a meta-algorithm that privately picks the mechanism with the better performance guarantee,\footnote{Which guarantee is best depends on the specific instance at hand.} and achieves performance as good as that of the best of the first two mechanisms.

Finally, in Subsection \ref{sec:lowerbound}, we show matching lower bounds (up to log factors) for the payoff of \emph{any} $(\epsilon,\delta)$-joint differentially private mechanism for the call auction.

\subsection{A Private Call Auction Mechanism via Coin Flipping}\label{subsec:coin}

In this subsection, we introduce our first jointly differentially private algorithm for selecting a price and allocating buyers and sellers to trades. The algorithm uses the exponential mechanism to differentially privately select a clearing price. With a slight abuse of notation, let
\begin{equation}\label{eq:pi_abuse}
\Pi \left(p, \vs, \vb \right) = \min \left\{\sum_{i \in \S} \1 \left[\vs_i \leq p\right], \sum_{j \in \B} \1 \left[\vb_j \geq p\right] \right\}
\end{equation}
be the number of trades that can happen at price $p$ while guaranteeing every agent non-negative utility; $\Pi \left(p, \vs, \vb \right)$ is the utility function used by the exponential mechanism. After choosing the price, the mechanism randomly selects buyers and sellers willing to transact at the chosen price by flipping a coin with some particular bias for every agent in the market. The exchange then transacts with all selected transactors, possibly taking a net position in the asset. We  formalize this mechanism in Algorithm \ref{alg:mech}. The mechanism takes data set $(\vs, \vb)$, privacy parameter $\epsilon$, and confidence parameter $\alpha$ as inputs and outputs a price $p$ and allocation vectors $\oo = (\os, \ob)$. In the algorithm description, $\exp(\cdot)$ is the exponential function, $Lap (\sigma)$ represents a mean-zero Laplace random variable with scale parameter $\sigma$, $(x)_+ := \max (x,0)$, and $Bern(q)$ represents a Bernoulli random variable with success probability $q$.

\begin{algorithm}
    \SetAlgoNoLine
    \KwIn{Agents' valuations $(\vs,\vb)$, privacy level $\epsilon$, confidence level $\alpha$.}
    \KwOut{Market price $p$, allocations $\oo=(\os,\ob)$.}
    
	Draw $p \propto \exp \left( \frac{\epsilon \Pi (p,\vs,\vb)}{2} \right)$ \Comment{Exponential mechanism chooses a price $p$ privately}
	
    $\widehat{s} \gets   \sum_{i \in \S}  \mathbf{1}\left[p \geq \vs_i\right] + Lap(\frac{1}{\epsilon}) $ 
    \Comment{Privately estimate \# of sellers willing to trade at $p$}
    
	$\widehat{b} \gets \sum_{j \in \B} \mathbf{1} \left[p \leq \vb_j \right] + Lap(\frac{1}{\epsilon})$ 
	\Comment{Privately estimate \# of buyers willing to trade at $p$}
	
	$\os_i \gets \mathbf{1} \left[ p \ge \vs_i \right] \cdot Bern \left( q^s = \min\left\{1, \frac{ \left( \widehat{b} \right)_+ }{ \left(\widehat{s} - \frac{\ln (1/\alpha)}{\epsilon} \right)_+ } 	\right\} \right)$ for all $i \in \S$.
	\Comment{Sellers' allocations}
	
	$\ob_j \gets \mathbf{1} \left[ p \le \vb_j \right] \cdot Bern \left( q^b = \min\left\{1, \frac{\left( \widehat{s} \right)_+ }{\left( \widehat{b} - \frac{\ln (1/\alpha)}{\epsilon} \right)_+ } \right\} \right)$ for all $j \in \B$.
	\Comment{Buyers' allocations}
	
	\caption{Private Call Auction with Allocation via Coin Flipping ($\M_1$)}
	\label{alg:mech}
\end{algorithm}

We start the analysis by providing the privacy guarantees obtained by Algorithm~\ref{alg:mech}:
\begin{clm}\label{clm:DP_1}
	The allocation mechanism described in Algorithm \ref{alg:mech} satisfies $3\epsilon$ joint differential privacy.
\end{clm}

The full proof of this claim can be found in Appendix \ref{app:DPproofs}. We also provide bounds on the payoff and the inventory of Mechanism~\ref{alg:mech} below:

\begin{thm}[Payoff and Inventory of Mechanism \ref{alg:mech}]\label{thm:thm_1}
Suppose $\opt \geq 5 \ln (V/\alpha)/\epsilon$.
\begin{enumerate}
\item Payoff: with probability $1-8\alpha$,
\begin{align*}
\Pi \left( \M_1 \right) 
&\ge \opt - \frac{2 \ln (V/\alpha)}{\epsilon} - \frac{2 \ln \left( 1/\alpha \right)}{\epsilon}  - \sqrt{6  \left(\opt + \frac{\ln (1/\alpha)}{\epsilon} \right) \ln \left(1/\alpha\right)}
\end{align*}
\item Inventory: with probability $1-6\alpha$,
$$
I \left( \M_1 \right) \le \frac{18 \ln (1/\alpha)}{\epsilon} + 2\sqrt{6 \left(\opt + \frac{\ln (1/\alpha)}{\epsilon} \right) \ln (2/\alpha) } + \frac{4 \ln (2/\alpha)}{3}
$$
\end{enumerate}
\end{thm}

\begin{rmk}\label{rmk:c}
Note that we constraint $\opt = \Omega \left( \ln \left( V/\alpha \right) / \epsilon\right)$. When $\opt = O \left( \ln \left( V/\alpha \right) / \epsilon\right)$, the inaccuracy introduced by releasing a differentially private price via the exponential mechanism is on the order of $\opt = \Omega \left( \ln \left( V/\alpha \right) / \epsilon \right)$, and we cannot hope to recover non-trivial utility guarantees.
\end{rmk}

The proof of Theorem~\ref{thm:thm_1} is given in Appendix~\ref{app: mechanism1}. 
We note that our bound does not follow directly from the classical guarantees of the Laplace and exponential mechanisms; it requires a more involved analysis of the concentration of the distribution of buyers and sellers selected to trade in Algorithm \ref{alg:mech}. 

%% file: our_mechanism_lottery.tex

\subsection{A Private Call Auction Mechanism via Lottery Numbers}
\label{subsec:lottery}

Here, we present a second mechanism that, rather than using independent randomization to decide who participates in a trade, uses correlated randomization to improve the payoff and reduce the inventory requirements of the mechanism. In the second mechanism, participants are given data-independent ``lottery numbers'', and thresholds on these lottery numbers (selected using the exponential mechanism) are used to select among willing traders on both sides of the market. This correlation allows us to remove the $\sqrt{\opt}$ term in the bounds of the previous mechanism, at the cost of introducing a logarithmic dependence on the number of agents $n$.

For a given valuation profile $(\vs,\vb)$, let $\Pi \left(p, \vs, \vb \right)$ be defined as in Equation \ref{eq:pi_abuse}. Assume seller $i$ is assigned a lottery number $l^s_i \in [n^s]$ and buyer $j$ is given $l^b_j \in [n^b]$ where we require that these lottery numbers are different for different agents. Without loss of generality, we assume $l^s_i = i$ and $l^b_j =j$. For a given price $p$, and profiles $(\vs,\vb)$, the loss of thresholds $\tau^s$ and $\tau^b$ on lottery numbers (one for sellers and one for buyers) is expressed as follows:
\begin{align*}
&L^s \left(\tau^s, p, \vs, \vb \right) = \left| \sum_{i \in \S}  \mathbf{1}\left[p \geq \vs_i , \tau^s \ge i \right]  - \Pi \left(p, \vs, \vb \right) \right|,
\\&L^b \left(\tau^b, p, \vb, \vb \right) = \left| \sum_{j \in \B}  \mathbf{1}\left[p \le \vb_j , \tau^b \le j \right]  - \Pi \left(p, \vs, \vb \right) \right|
\end{align*}
For a price $p$, these loss functions measure how far off the number of agents chosen to trade on each side of the market would be from our target number of trades, $\Pi(p, \vs, \vb)$, if we used thresholds $\tau^s$ and $\tau^b$ as a tie-breaking rule to select sellers and buyers who are willing to trade at price $p$, respectively. In Algorithm \ref{alg:mech_lottery}, just as before, we first use the exponential mechanism to select a price and then use the exponential mechanism with loss functions $L^s$ and $L^b$ (or utility functions: $-L^s$ and $-L^b$, based on the terminology used to describe the exponential mechanism in Appendix \ref{app:dp_tools}) to select the thresholds on lottery numbers.

\begin{algorithm}
    \SetAlgoNoLine
	\KwIn{Agents' valuations $(\vs,\vb)$, privacy level $\epsilon$.}
    \KwOut{Market price $p$, allocations $\oo=(\os,\ob)$.}
    
    Draw $p \propto \exp \left( \frac{\epsilon \Pi (p, \vs, \vb)}{2} \right)$ \Comment{Exponential mechanism to privately choose a price $p$}
    
	Draw $\tau^s \propto \exp \left( - \frac{\epsilon L^s(\tau^s, p, \vs, \vb)}{4} \right)$ \Comment{Exponential mechanism to privately choose $\tau^s$}
	
    Draw $\tau^b \propto \exp \left( - \frac{\epsilon L^b(\tau^b, p, \vs, \vb)}{4} \right)$ \Comment{Exponential mechanism to privately choose $\tau^b$}
    
	$\os_i \gets \mathbf{1} \left[ p \ge \vs_i , \tau^s \ge i  \right]$ for all $i \in \S$.
	\Comment{Sellers' allocations}
		
	$\ob_j \gets \mathbf{1} \left[ p \le \vb_j, \tau^b \le j \right] $ for all $j \in \B$.
	\Comment{Buyers' allocations}
	
	\caption{Private Call Auction with Allocation via Lottery Numbers ($\M_2$)}
	\label{alg:mech_lottery}
\end{algorithm}

\begin{clm}\label{clm:DP_2}
	The allocation mechanism described in Algorithm \ref{alg:mech_lottery} satisfies $3\epsilon$ joint differential privacy.
\end{clm}

\begin{thm}[Payoff and Inventory of Mechanism \ref{alg:mech_lottery}]\label{thm:thm_2}
For any $\alpha > 0$,
\begin{enumerate}
\item Payoff: with probability $1-3\alpha$,
\begin{align*}
\Pi \left( \M_2 \right) 
&\ge \opt - \frac{2 \ln \left(V/  \alpha \right)}{\epsilon} - \frac{4 \ln \left(n/ \alpha \right)}{\epsilon} 
\end{align*}
\item Inventory: with probability $1-2\alpha$,
$$
I \left( \M_2 \right) \le \frac{8 \ln \left(n/ \alpha \right)}{\epsilon},
$$
\end{enumerate}
\end{thm}

The proof of Claim \ref{clm:DP_2} is provided in Appendix \ref{app:DPproofs}, and that of Theorem \ref{thm:thm_2} in Appendix~\ref{app:mechanism2}.

\subsection{A Meta Algorithm: Selecting the Best Mechanism Privately}
\label{subsec:meta}

Notice that the first term in the payoff bounds of both Theorems \ref{thm:thm_1} and \ref{thm:thm_2} are identical (as they both correspond to choosing a price using the exponential mechanism) but the remaining terms differ ($\M_1$ relies on binomial coin flips for tie-breaking whereas $\M_2$ tie-breaks via thresholds on lottery numbers). These two bounds are in general not comparable, as one depends on the maximum number of shares $\opt$ that can be cleared, whereas the other one depends on the total number $n$ of agents in the market. The first bound provides better guarantees (up to constants and $\ln \left(1/\alpha \right)$ terms) when $\sqrt{\opt} < \ln \left(n \right) / \epsilon$, i.e. when the number of possible trades is significantly smaller than the total number of agents in the market,\footnote{This models practical situations in repeated financial markets where sellers price a security higher than most buyers are willing to pay. In such situations, buyers may elect to wait until a new seller comes and offers a better price, while sellers may wait for a new buyer willing to buy at the current price.} whereas the second bound provides better guarantees when $\sqrt{\opt} > \ln \left(n \right)/ \epsilon$.

We can achieve the better of the two bounds by comparing the bounds of Theorems \ref{thm:thm_1} and \ref{thm:thm_2} in a differentially-private manner and then running the mechanism with the better bound according to this private computation. 
To do so we compute the difference of payoff bounds of Mechanisms~\ref{alg:mech} and~\ref{alg:mech_lottery}
$$
f \triangleq \frac{2 \ln \left( 1/\alpha \right)}{\epsilon}  + \sqrt{ 6  \left(\opt + \frac{\ln (1/\alpha)}{\epsilon} \right) \ln \left(1/\alpha\right)} - \frac{4 \ln (n/ \alpha)}{\epsilon}
$$
in a differentially privately manner.\footnote{$\opt$ is a function of the input data set, hence a direct comparison of the bounds without addition of noise may leak information about the reported bids.} Then, based on the sign of $f$, the mechanism decides whether to run $\M_1$ or $\M_2$. The private computation of $f$ will unavoidably add an extra term of order $\bigo (1/\epsilon)$ to the final payoff bound. This mechanism is described in Algorithm~\ref{alg:mech_combined}. We provide guarantees on privacy, payoff, and inventory of this mechanism below.

\begin{algorithm}
    \SetAlgoNoLine
	\KwIn{Agents' valuations $(\vs,\vb)$, privacy level $\epsilon$, confidence level $\alpha$.}
    \KwOut{Market price $p$, allocations $\oo=(\os,\ob)$.}
	$f \gets \frac{2 \ln \left( 1/\alpha \right)}{\epsilon}  + \sqrt{ 6  \left(\opt + \frac{\ln (1/\alpha)}{\epsilon} \right) \ln \left(1/\alpha\right)} - \frac{4 \ln (n/ \alpha)}{\epsilon}$
	
	$\widetilde{f} \gets f + Lap \left(\frac{\sqrt{6 \ln(1/\alpha)}}{\epsilon} \right)$ \Comment{Private estimate of $f$}
	
	\eIf{$\widetilde{f} < 0$}
	{
		Run $\M_1 \left(\vs,\vb,\epsilon, \alpha \right)$ (Algorithm \ref{alg:mech}) and get $p,\oo$.
    }
    {
		Run $\M_2 \left(\vs,\vb,\epsilon \right)$ (Algorithm \ref{alg:mech_lottery}) and get $p,\oo$.
	}
	\caption{Private Call Auction with Allocation: A Meta Algorithm ($\M_3$)}
	\label{alg:mech_combined}
\end{algorithm}

\begin{clm}\label{clm:DP_3}
	The allocation mechanism described in Algorithm \ref{alg:mech_combined} satisfies $7\epsilon$ joint differential privacy.
\end{clm}
The proof of Claim \ref{clm:DP_3} is provided in Appendix \ref{app:DPproofs}.

\begin{thm}[Payoff and Inventory of Mechanism \ref{alg:mech_combined}]\label{thm:thm_3}
Suppose $\opt \geq 5 \ln (V/\alpha)/\epsilon$.
\begin{enumerate}
\item Payoff: with probability $1-18\alpha$,
\begin{align*}
&\Pi \left( \M_3 \right)  \\
&\ge \opt - \frac{2 \ln (\frac{V}{\alpha})}{\epsilon} - \min\left\{ \frac{2 \ln \left( \frac{1}{\alpha} \right)}{\epsilon}  + \sqrt{ 6  \left(\opt + \frac{\ln (\frac{1}{\alpha})}{\epsilon} \right) \ln \left(\frac{1}{\alpha}\right)},  \frac{4 \ln (\frac{n}{\alpha})}{\epsilon} \right\} - \frac{\sqrt{6} \ln^{1.5} (\frac{1}{\alpha})}{\epsilon}
\end{align*}

\item Inventory: with probability $1-14\alpha$,
\begin{align*}
&I \left( \M_3 \right) \\
&\le 4 \min\left\{ \frac{2 \ln \left(\frac{1}{\alpha} \right)}{\epsilon}  + \sqrt{ 6  \left(\opt + \frac{\ln \left(\frac{1}{\alpha}\right)}{\epsilon} \right) \ln \left(\frac{1}{\alpha}\right)},  \frac{4 \ln (\frac{n}{\alpha})}{\epsilon} \right\} +  \frac{4 \sqrt{6} \ln^{1.5} (\frac{1}{\alpha})}{\epsilon} + \frac{10 \ln (\frac{1}{\alpha})}{\epsilon} + \frac{4 \ln (\frac{2}{\alpha})}{3}
\end{align*}

\end{enumerate}
\end{thm}

This theorem follows from Theorems \ref{thm:thm_1} and \ref{thm:thm_2}, as well as the accuracy guarantee of the Laplace mechanism used in Algorithm \ref{alg:mech_combined} to compute $f$. We defer the full proof to Appendix~\ref{app:mechanism3}.

%% file: lowerbound.tex
\subsection{A Lower Bound}\label{sec:lowerbound}
We now provide a lower bound showing that \emph{any} algorithm which computes a price in an $(\epsilon,\delta)$-differentially private manner and allocates among willing participants at this price \emph{must}, for \emph{some} instance, suffer a loss of $\Omega\left(1/\epsilon\right)$ (compared to the optimal number of shares that could be cleared on that instance). 
Because this bound applies to a broader set of mechanisms that reveal \emph{only} the price privately (but may select the optimal allocation absent privacy), it also applies to the mechanisms considered in Section~\ref{sec:mechanisms}.
We will compare the performance of any given differentially private algorithm on several input data sets. To do so, we will define an instance-dependent benchmark below, that we call $\opt(D)$. Formally, given an input data set $D = (\vs,\vb)$, our benchmark is:
\[
\opt \left(D \right) = \max_p \min \left\{\sum_{i \in \S} \1 \left[\vs_i \leq p \right], \sum_{j \in S} \1 \left[\vb_j \geq p \right] \right\}.
\]

\begin{definition}[Loss of an algorithm]
For any (possibly randomized) algorithm $\A : \D^n \to P$ that takes a data set $D =(\vs,\vb)$ as an input and outputs a price $p$, the loss of $\A$ on input data set $D = (\vs,\vb)$ of agents valuations is defined as follows:
$$
L \left(\A,D \right) = 
\opt \left(D \right)
- \E_{p \sim \A(D)} \left[ \min \left\{\sum_{i \in \S} \1 \left[\vs_i \leq p \right], \sum_{j \in S} \1 \left[\vb_j \geq p \right] \right\} \right].
$$
I.e., this loss compares the number of trades that could be cleared in expectation at the price selected by $\A$ 
to the maximum number of trades when the trading price is optimally chosen. We define the worst-case expected loss of $\A$ as the worst-case loss over all data sets, i.e. $L \left(\A \right) = \sup_D \left[ L (\A,D) \right]$.
\end{definition}

Our lower bound will hold so long as $\delta$ is not too large in comparison with $\epsilon$.\footnote{Typically, differentially private algorithms use $\delta << \epsilon$.} We note that our lower bound on the expected loss matches the $\widetilde{\bigo}\left(1/\epsilon\right)$ dependencies\footnote{The instances we construct use $V,n \sim 1/\epsilon$. The logarithmic dependencies of our upper bounds in $n$ and $V$ translate into logarithmic dependencies in $1/\epsilon$, hence the $\widetilde{\bigo}$ notation.} of our high probability upper bounds on the loss for Mechanisms~\ref{alg:mech},~\ref{alg:mech_lottery},~\ref{alg:mech_combined} (and consequently of any upper bound on the expected loss of these mechanisms). Finally, it is worth remarking that our lower bound for $(\epsilon,\delta)$-DP mechanisms matches the upper bound obtained by restricting attention to $(\epsilon,0)$-DP mechanisms; this implies that relaxing $\delta$-privacy requirements of Mechanisms~\ref{alg:mech},~\ref{alg:mech_lottery},~\ref{alg:mech_combined} will not lead to any significant improvements in terms of their accuracy guarantees. 

\begin{thm}\label{thm:lowerbounds}[Lower bound on the loss of private algorithms] Pick any $\epsilon, \delta$ such that  $0 \le \epsilon \le 1$ and $\delta = \bigo (\epsilon)$. There exists a range of (integer) valuations $P(\epsilon)$ and a number of agents $n(\epsilon)$ such that \emph{any} $(\epsilon, \delta)$-DP algorithm $\A : \D^{n(\epsilon)} \to P(\epsilon)$ must suffer worst-case expected loss of $\Omega\left(1 /\epsilon \right)$.
\end{thm}

The proof of Theorem~\ref{thm:lowerbounds} relies on constructing a family of data sets $\{D_l\}_l$ such that no differentially private algorithm $\A$ can simultaneously suffer expected loss of $\bigo \left(1 / \epsilon\right)$ on all of them. We do so by carefully calibrating the following trade-off: on the one hand, we require any pair of data sets in $\{D_l\}_l$ be close enough that the stability properties of differential privacy guarantee any private algorithm must pick a similar distribution of prices on both data sets. On the other hand, we require that the data sets furthest from each other are different \emph{enough} such that no fixed distribution can incur a low loss on both.

%% file: connections_to_finance_lit.tex
\subsection{Connections to the Market Impact Literature}\label{sec:marketimpact}

As mentioned in the Introduction, it is possible to draw some informal but interesting connections
between this work and the finance literature
on market impact. Market impact models typically propose strong
stochastic assumptions on price formation (e.g. random walk and diffusion models or martingale assumptions on limit order dynamics) and then solve for the optimal strategy to
minimize trading costs and price impact. In particular, there is a large body of work
on the so-called ``square root law'' (see, eg. \cite{gatheral2010no,gatheralslides}). 
which predicts that the change to price
inflicted by a trade of $k$ shares scales with $\sqrt{k/\mathcal{V}}$, where $\mathcal{V}$ is the total
volume of shares cleared during the trade; the ratio $k/\mathcal{V}$ is referred to as the 
trade's {\em participation rate}. As we note below, $\mathcal{V}$ is typically closely related to
other measures of market activity such as the number of orders placed (as with our $n$)
or the number of quote changes in limit order dynamics.

Our results imply that the change in the expected clearing price in our DP call auction
resulting from an order of $k$ shares is bounded by a multiplicative
factor of  $(e^{k\epsilon} - 1)$. Setting this equal to $\sqrt{k/n}$ to match the square root law\footnote{Here we are assuming that the number
of orders $n$ in our model plays the role of $\mathcal{V}$ above; see subsequent footnote.}
and solving for $\epsilon$ approximately yields 
$\epsilon \approx 1/\sqrt{kn}$ for small participation rates.
Plugging this into our utility bound of Theorem \ref{thm:thm_2}, the shares
we execute at this $\epsilon$ scales like $\opt(1 - \sqrt{kn}/\opt)$. 
Thus as long as $k$ is
$o \left(n \right)$ and $\opt$ scales with $n$,\footnote{This scaling is broadly consistent with recent
data from electronic exchanges. For instance, the ratio of 
shares traded to 
quote changes (a common measure of market activity) across 3443 U.S. equities averaged 0.16
with 
standard deviation 
0.09.}
asymptotically we approach $\opt$ with the same price impact as
that predicted by the square root law but with two major advantages.
First, we have made {\em no} assumptions, stochastic or otherwise,
on the orders placed by market participants. Second, we are not only bounding
the price impact, we are also bounding information leakage of {\em any} form, as
per the promises of differential privacy.

%% file: GAME-SEC.tex
\section{Strategic Framework}\label{sec:game}

In Section \ref{sec:mechanisms}, we focused on the algorithmic form of our mechanism and provided privacy guarantees and optimality guarantees with respect to the reported valuations, without regard to whether those reports are truthful or not. In this section, we embed our mechanisms into a game theoretic framework and examine its properties, including (approximate) truthfulness. More precisely, we now assume the agents are strategic; they may decide to report a bid that differs from their valuation, or even to not participate in the mechanism in the first place. Formally, all sellers $i$ and buyers $j$ have quasi-linear utilities determined by their own valuations and the outcome of the mechanism:
$
\us_{i}(\M) = \os_i \cdot (p - \vs_i)
,
\ub_{j}(\M) = \ob_i \cdot (\vb_j - p),
$
where, with a slight abuse of notation, we omit the dependency of $\M$ on the agents' reports.

Buyers and sellers aim to maximize their utility from participating (or not participating) in the mechanism. In the face of strategic behavior, we will require our mechanisms to be (approximately) truthful and individually rational; i.e., it should never be in an agent's best interest to misreport his valuation, and an agent should always have a strategy that guarantees non-negative utility from participating in the mechanism and so would rather participate than not. Individual rationality and (approximate) truthfulness are formally defined below:

\begin{definition}[Ex-Post Individual Rationality]\label{def:IR}
We say a double-auction mechanism $\M$ satisfies ex-post individual rationality if, for every seller $i \in \S$, there exists a bid $\ts_i$ for agent $i$ such that for every possible set of bids $\mathbf{r}_{-i}$ submitted by all agents but $i$, and every realization of the randomness of the mechanism $\M$,
$
\us_i \left(\M \left( \ts_i, \mathbf{r}_{-i} \right)\right)  \geq 0
$,
and similarly for every buyer $j$, there exists a bid $\tb_j$ for agent $j$ such that for every possible set of bids $\mathbf{r}_{-j}$ submitted by all agents but $j$, and every realization of the randomness of the mechanism $\M$,
$
\ub_j \left(\M \left( \tb_j, \mathbf{r}_{-j} \right) \right) \geq 0.
$
\end{definition}

\begin{definition}[Approximate Dominant-Strategy Truthfulness]\label{def:truthfulness}
We say a double-auction mechanism $\M$ satisfies $\gamma$-approximate dominant-strategy truthfulness if, for every seller $i \in \S$, every possible bid
$\ts_i$ submitted by $i$, and every possible set of bids $\mathbf{r}_{-i}$ submitted by all agents but $i$,
$$
\E_\M \left[ \us_i \left(\M \left( \ts_i, \mathbf{r}_{-i} \right) \right)\right] \le \E_\M \left[ \us_i \left(\M \left( \vs_i, \mathbf{r}_{-i} \right) \right)\right] + \gamma
$$
and similarly for every buyer $j$, for every possible bid
$\tb_j$ submitted by $j$ and every possible set of bids $\mathbf{r}_{-j}$ submitted by all agents but $j$,
$$
\E_\M \left[ \ub_j \left(\M \left( \tb_j, \mathbf{r}_{-j} \right) \right)\right] \le \E_\M \left[ \ub_j \left(\M \left( \vb_j, \mathbf{r}_{-j} \right) \right)\right] + \gamma
$$
where expectations are taken with respect to the randomness of $\M$.
\end{definition}

In Section \ref{sec:iric}, we show that our mechanisms are \emph{individual rational} and (unlike in the standard call auction) approximately \emph{dominant-strategy truthful}. While our results assume that agents wish to trade a single share, we show how our per-share guarantees translate into (gracefully degrading) per-player guarantees in more general setting in which agents can trade multiple shares.

Then, in Sections \ref{sec:learning_noprivacy}-\ref{sec:learning_privacy}, we consider \emph{learning dynamics} under both the standard call auction and our mechanism and show that a system in which agents use a modified exponential weights algorithm (which we call ``Social'' Exponential Weights) to learn to bid will eventually converge to the optimal number of shares cleared. While it is true that truthfulness implies that agents cannot do better than bidding their true values, one might consider learning dynamics for two reasons. First, if agents do not trust the mechanism designer (or share their assumptions), applying a no-regret learning algorithm is a plausible response to guarantee good performance. Second, good outcomes obtained in the presence of decentralized, distributed, and selfish algorithms are compelling evidence of the robustness and quality of our mechanism. 
To our knowledge, the use of no-regret learning algorithms by all agents in a call auction setting has not been studied before, and these results may be of independent interest.

\subsection{Individual Rationality and Truthfulness Properties of Our Algorithms} \label{sec:iric}
In this section, we discuss the incentive properties of our proposed algorithms. To do so, we note that $(\vs, \vb)$ are the true valuations of sellers and buyers and denote their revealed bids by $(\ts,\tb)$, respectively. To study truthfulness, we assume seller $i$ (buyer $j$) can submit a bid $\ts_i$ ($\tb_j$) that may not be equal to their valuation $\vs_i$ ($\vb_j$), and show that it is approximately never in agent $i$'s (resp. $j$'s) best interest to do so.
We start by noting that our mechanisms are individually rational:
\begin{clm}[Individual Rationality] \label{clm:IR}
 The mechanisms described in Algorithms \ref{alg:mech}, \ref{alg:mech_lottery}, and \ref{alg:mech_combined} are ex-post individually rational. 
\end{clm}
\begin{proof}
We prove the result for Mechanism \ref{alg:mech}; proofs for the other mechanisms are similar. It suffices to show that there exists a strategy for any seller (resp. any buyer) that guarantees him non-negative utility. For any seller $i$, setting $\ts_i = \vs_i$  is a strategy that ensures that whenever $i$ is allocated a trade (i.e. $\os_i = 1$), it must be that $p \geq \ts_i = \vs_i$; this immediately guarantees $i$ gets non-negative utility---independently of how other agents bid and of the randomness of the mechanism. A similar proof holds for buyers. 
\end{proof}

We also show that differential privacy guarantees approximate truthfulness in the dominant-strategy sense: i.e., it does not allow agents (sellers and buyers) to gain too much profit by submitting a bid different than their true valuation, \emph{no matter what the realized bids of the other agents are}\footnote{Truthfulness is desirable not only because it makes computing equilibrium strategies and predicting equilibrium behavior simpler, but also because knowing the \emph{true} valuations allows the mechanism designer to clear the most shares. 
}. 

\begin{clm}[Approximate Truthfulness]\label{clm:appxic}
	The mechanisms described in Algorithms \ref{alg:mech} and \ref{alg:mech_lottery} satisfy $\gamma$-approximate dominant-strategy truthfulness for $\gamma = (e^{3\epsilon}-1)V$; the mechanism described in Algorithm \ref{alg:mech_combined} satisfies $\gamma$-approximate dominant-strategy truthfulness for $\gamma = (e^{7\epsilon} -1) V$.
\end{clm}

We defer the full proof to Appendix~\ref{app:IC-IR}. In the proof, we first observe that since the market price is chosen subject to differential privacy, individual agents cannot significantly change it by misreporting their valuations. However, this is not enough to argue truthfulness, as under \emph{joint} differential privacy, an agent's allocation may heavily depend on his report. To complete the proof, we show that the function by which the mechanism determines transactors is a best-response for an agent with the reported valuation given the output of the differentially private mechanism. 

Note that, in general, call auctions are \emph{not} dominant-strategy truthful, since even small bidders may impact the price selected by a mechanism acting on reported bids. This is a consequence of the fact that in the simple call auction (as well as in continuous order book mechanisms) the optimal price is, in general, \emph{not} stable \cite{michael}. 
Importantly, we note that the truthfulness guarantees are a function of $\epsilon$; as $\opt$ grows larger, $\epsilon$ can be made smaller with less and less relative cost. Consequently, the truthfulness guarantee can be made stronger for a given level of privacy as the number of optimal trades cleared increases. 

We highlight that because our strategic framework assumes each bidder controls a single share, our guarantees are at the \emph{per-share} level. Our privacy guarantees generalize, however, to the case where bidders control at most $k$ shares by expanding $\epsilon$ by a factor of $k$.\footnote{An $\epsilon$-differentially private mechanism with respect to a single share is $k \epsilon$-differentially private with respect to the data of a bidder who controls $k$ shares; intuitively, this is because an agent that misreports his valuation over $k$ shares creates a dataset that is a $k$-neighbor of the dataset in which they had bid truthfully. Such a bidder can affect the distribution of prices by an amount of at most $e^{k \epsilon}$, and so their own expected utility by  $(e^{k \epsilon} - 1)V$.} Our truthfulness guarantees also follow by expanding $\epsilon$ by a factor of $k$.

%% file: noregret_dynamics.tex
\subsection{Learning in Repeated Call Auctions}

In this section, we consider a \emph{repeated} 
call auction. Agents are initially unaware of each other's valuations and behavior and run simple learning algorithms to learn how to bid. In each time step $t$, each seller $i$ (respectively buyer $j$) reports a bid $\ts_{i,t}$ (resp. $\tb_{j,t})$, which may differ from his valuation, to the mechanism. Given this input $(\ts_t,\tb_t)$, the mechanism computes and publicly releases a price $p_t$ and assigns an allocation $\oo_{i,t}$ to each seller $i$ (respectively $\oo_{j,t}$ to each buyer $j$). We will consider two versions of this mechanism, one that is non-private and is inspired by standard call auctions, in Section~\ref{sec:learning_noprivacy}, and one that is private and is based on Mechanism~\ref{alg:mech}, in Section~\ref{sec:learning_privacy}. The agents then update their bidding strategies based on the quantities outputted by the mechanism, via a simple no-regret algorithm (Exponential Weights). 

We highlight that our agents are \emph{naive} in that they do not compute a counterfactual price $p_t$ and allocation vector $\oo_t$ given alternative bids they could have made. Instead, they only update their bidding strategies with respect to how much better off they could have been by bidding differently, \emph{assuming they had no effect on the price}. The motivation for this is two-fold: first, counterfactual reasoning would require the agents to know the bids of other agents, which are not released by the mechanism (and typically not available in many real-life call auctions). Second, when agents are small relative to the total market, they may believe that their actions \emph{do not} greatly affect these quantities. We note that differential privacy makes this belief into a \emph{property} of our mechanism rather than a naive assumption. Thus, small bidders using naive updates will have a \emph{real} regret guarantee when interacting with a differentially private call auction.

\subsubsection{Learning in the Absence of Privacy}\label{sec:learning_noprivacy}
In this section, we focus on learning dynamics when the mechanism runs a standard call auction, absent privacy; this non-private setting will serve as a natural point of comparison for dynamics with respect to our private mechanism. At every time step $t$, agents submit bids that may differ from their valuations. In response, the mechanism computes a price and allocation, with the goal of maximizing traded shares among willing participants. We denote the agents' reports as $\ts_{i,t}$ and $\tb_{j,t}$ for seller $i$ and buyer $j$, respectively, at time $t$. The mechanism chooses a price $p_t$ to maximize
\[
\Pi \left(p,\ts_t,\tb_t \right) \triangleq \min\left\{\sum_{i \in \S} \1 \left[\ts_{i,t} \leq p \right],\sum_{j \in \B} \1 \left[\tb_{j,t} \geq p \right]\right\},
\]
which is the number of shares the mechanism will trade at price $p$, assuming that sellers will only agree to trade when the price is above their reported bid and buyers when the price is below their reported bid. To compute the allocation $\oo_t = (\os_t,\ob_t)$, the mechanism must choose among these sellers and buyers it believes (based on the reports) are willing to trade at the chosen price $p_t$. When there are an equal number of sellers and buyers willing to trade at price $p_t$, the mechanism allocates a trade to all of them; otherwise, the mechanism randomly selects a subset of $\Pi (p,\ts_t,\tb_t)$ agents from the side with excess number of willing participants. Formally, the mechanism computes $q_t^b = \Pi(p_t,\ts_t,\tb_t)/ \sum_{j \in \B} \1[\tb_{j,t} \geq p_t]$ and $q^s_t = \Pi(p_t,\ts_t,\tb_t) / \sum_{i \in \S} \1 [\ts_{i,t} \leq p_t]$; these probabilities will be less than $1$ on the excess side of the market and exactly $1$ on the short side. We assume the mechanism publicly releases $p_t$, $q_t^s$, and $q_t^b$ to all agents in the market, and communicates to each seller $i$ (resp. buyer $j$) his own allocation $\os_{i,t}$ (resp. $\ob_{j,t}$).

\paragraph{Agents learn via Exponential Weights:}

A natural no-regret (\emph{regret} here is the classic notion of performance in online learning) algorithm for updating bidding strategies is the Exponential Weights mechanism. We describe the classical \emph{Exponential Weights Update} rule for buyers (buyer $j$) in Algorithm~\ref{alg:mw}, and note that this update is defined symmetrically for the sellers.

\begin{algorithm}
    \SetAlgoNoLine
	\KwIn{Learning rate $\eta$.}
	Set $\wb_{j,1}(k)\gets \frac{1}{\vb_j } $ for $k = 1,\ldots, \vb_j$ 
	\Comment{Initialize uniform weights.}
	
	\For{$t \in 1...T$}
	{
        $\tb_{j,t} \sim \wb_{j,t}$
        \Comment{Draw bid from distribution}
        
		$\umodb_{j,t}(k) \gets q_t^b(\vb_j-p_t) \mathbf{1}[k \geq p_t]$ for $k = 1,\ldots, \vb_j$ 
		\Comment{Observe payoff of each bid $k$}
		
	    $\wb_{j,t+1}(k) \gets \frac{\exp(\eta \umodb_{j,t}(k))}{\sum_{j} \wb_{j,t}(j) \exp(\eta \umodb_{j,t}(j))} \cdot  \wb_{j,t}(k)$ for $k = 1,\ldots, \vb_j$ 
	    \Comment{Update the weights.}
	}
	\caption{Exponential Weights}
	\label{alg:mw}
\end{algorithm}

Informally, the updates work as follows. Initially, we assume every seller bids uniformly above their value and every buyer bids uniformly below their value.\footnote{A buyer $j$ cannot improve his utility by bidding over his valuation (as increasing his bid cannot decrease $p_t$ nor increase his probability of allocation), and risks obtaining negative utility by doing so, if $\vb_{j,t} < p_t \leq \tb_{j,t}$. Hence, bidding above his valuation is a dominated strategy for the buyer. Similarly, bidding under his value is a dominated strategy for a seller. The assumption of this prior knowledge can be relaxed at the price of slower convergence.} 
Then, in each round  $t$, for every possible $k \in P$, agents compute what their expected payoff would have been had they reported $k$ as their valuation, given the \emph{current} price $p_t$ and the allocation probabilities $q_t^s$ and $q_t^b$. They use these expected payoffs to update their distribution of bids, in a way that puts exponentially more weight on bids with higher expected utilities; the speed at which these updates happen is controlled by the learning rate parameter $\eta$, taken here to be constant. For appropriate choices of learning rate $\eta$, this algorithm is known to be no-regret.

One may hope these dynamics converge to clearing $\opt$ shares with probability going to $1$, where $\opt$ is defined as in Equation \ref{eq:opt}. However, this may not be the case when agents update their weights according to Algorithm~\ref{alg:mw}. This stems from the fact that agents are indifferent between trading at their valuation, and not trading at all, as both net a payoff of zero. This is reflected in the exponential weight update, and buyers learn to put a significant amount of weight on bids that are strictly less than their valuation (as trades for those bids are strictly profitable). When clearing $\opt$ trades requires many agents to bid exactly at their valuation, the number of shares cleared is bounded away from the benchmark. We show instead that the dynamics will clear the following benchmark, which only considers trades that are strictly profitable for both sides of the market:
\begin{definition}[Optimal Jointly Profitable Trades]\label{def:opt'}
	 We let $\opt'$ be the maximum number of trades achievable for a given $(\vs,\vb)$, such that all trading buyers and sellers get strictly positive utility. Formally, 
	 \[
	 \opt' = \max_{p} \min \left\{\sum_{i \in \S} \1 \left[\vs_i < p\right], \sum_{j \in \B} \1 \left[\vb_j > p\right] \right\}.
	 \]
	 We call this benchmark the ``Optimal Jointly Profitable Trades with Uniform Pricing'' benchmark.
\end{definition}
The statement showing that the mechanism will converge in probability to clearing at least $\opt'$ shares is formalized below:

\begin{thm}[Convergence to (at least) $\opt'$]\label{thm:limit_1}
Suppose buyers and sellers update their bid distributions according to Algorithm \ref{alg:mw} (with any $\eta > 0$). Further, at any time $t$, suppose $p_t$ is chosen uniformly at random among the set of optimal prices at time $t$. Then,
the number of shares cleared at time $t$ satisfies 
\begin{align*}
	\lim_{t \to \infty} \Pr \left[\Pi \left(p_t,\ts_t,\tb_t \right) \ge \opt' \right] = 1.
 \end{align*}

 \end{thm}

We also provide a variant of the Exponential Weights algorithm, that we will show converges to $\opt$ shares cleared. This variant is described in Algorithm~\ref{alg:mw_variant}. 

\begin{algorithm}
    \SetAlgoNoLine
	\KwIn{Learning rate $\eta$, ``fake'' utility $\xi$.}
	$\wb_{j,1}(k)\gets 1/\vb_j $ for $k = 1,\ldots, \vb_j$ 
	\Comment{Initialize with uniform weights.}
	
	\For{$t = 1, \ldots$, T}
	{
		$\tb_{j,t} \sim \wb_{j,t}$ 
		\Comment{Draw from the current weights.}
		
		\eIf{$\vb_j \neq p_t$}
		{
			$\umodb_{j,t}(k) \gets q_t^b(\vb_j-p_t) \mathbf{1}[k \geq p_t]$ for $k = 1,\ldots, \vb_j$ 
			\Comment{expected utility for bid $k$.}
		}
		{
			$\umodb_{j,t}(k) \gets q_t^b \xi \mathbf{1}[k = \vb_j]$  for $k = 1,\ldots, \vb_j$
			\Comment{Agent pretends getting utility $\xi$ from trading.}
		}		
		$\wb_{j,t+1}(k) \gets \frac{\exp(\eta \umodb_{j,t}(k))}{\sum_{l} \wb_{j,t}(l) \exp(\eta \umodb_{j,t}(l))} \cdot  \wb_{j,t}(k)$ for $k = 1,\ldots, \vb_j$
		\Comment{Update the weights.}		
	}
	\caption{Social Exponential Weights}
	\label{alg:mw_variant}
\end{algorithm}

Algorithm \ref{alg:mw_variant} is a modification of the classic Exponential Weights algorithm. In particular, when the price is equal to agent's valuation, the algorithm assigns a nonzero utility $q_t^b \xi$ to reporting the agent's valuation, for $\xi$ arbitrarily small; this can be seen as agents updating their weights as if they strictly preferred trading to not trading, even when their trade would make no profit. In other words, it implements a preference to break ties (in utility) in favor of trading over not trading. We call this ``Social'' Exponential Weights because incorporating this modified utility allows the system as a whole to reach a better social outcome (one with more shares traded) than otherwise. Crucially, despite this modification, Algorithm \ref{alg:mw_variant} remains no-regret for a fixed horizon $T$ with appropriate choices of learning rate, $\eta$, and ``fake'' utility, $\xi$:

\begin{lemma}[No-regret]\label{cor:noregret_variant}
Algorithm \ref{alg:mw_variant} is no-regret ($\bigo \left( \sqrt{T} \right)$ cumulative regret) for $\eta,\xi = \bigo\left(1/\sqrt{T}\right)$.
\end{lemma}

The proof is almost identical to that of the no-regret guarantees of traditional exponential weights, and is deferred to Appendix~\ref{app:noregret_variant}. We highlight that we define regret with respect to the single best action in hindsight given the \emph{fixed} sequence of prices observed; that is, we do not consider the notion of \emph{Stackelberg} regret, which is calculated with respect to the best fixed action \emph{given} that the mechanism picks a sequence of prices in response to the selected actions (see, e.g. \cite{dong2018strategic}). If agents are small enough that their actions do not greatly affect the mechanism's responses, then the standard notion of regret and Stackelberg regret do not greatly differ; if, moreover, a mechanism is differentially private, then (for small enough agents) these notions of regret coincide, because differential privacy ensures that agents placing small orders have little impact on the price. 

Under Algorithm~\ref{alg:mw_variant}, the number of shares cleared converges to $\opt$ with probability that tends to $1$ as $t$ grows large. We make this statement formally below:

\begin{thm}[Convergence to $\opt$]\label{thm:limit}
Suppose buyers and sellers that update their bidding strategies according to Algorithm~\ref{alg:mw_variant} (with any $\eta,\xi > 0$). Further, suppose $p_t$ is chosen uniformly at random among the set of optimal prices at time $t$. Then, the number of shares cleared at time $t$ satisfies 
\begin{align*}
	\lim_{t \to \infty} \Pr \left[\Pi \left(p_t,\ts_t,\tb_t \right) = \opt \right] = 1.
 \end{align*}
 \end{thm}

To prove this result, we show that with a small, constant probability (in $t$) in any given round, all agents bid their valuations. In such cases, the mechanism picks an optimal price, and at least $OPT$ buyers (resp. sellers) increase their probability of bidding above (resp. below) this price. When the number of rounds goes to infinity, this event is repeated infinitely often for some optimal price $p^\star$, and $OPT$ buyers (resp. sellers) bid above (resp. below) $p^\star$ with probability that tends to $1$. The full proof is given in Appendix~\ref{app:dynamics_notindifferent}. A similar argument is used to prove Theorem~\ref{thm:limit_1}, in Appendix~\ref{app:dynamics_indifferent}.

\subsubsection{Learning in Repeated Call Auctions with Differential Privacy}\label{sec:learning_privacy}

We now consider the same dynamic setting as before, with the difference that the centralized designer now computes the price $p_t$ and the allocation $\oo_t$ at time $t$ in a joint-differentially private fashion. For simplicity of exposition, we pick the private mechanism used by the designer to be Mechanism~\ref{alg:mech}, which picks a price via the exponential mechanism and picks agents to allocate from the smaller side of the market via binomial coin flips. We show that when agents play according to the exponential weights (resp. Social EW) algorithm, the dynamics converge to clearing at least $\opt$ (resp. $\opt'$) shares minus inaccuracies introduced by privacy.

\begin{thm}\label{thm:limit_private} Suppose buyers and sellers update their bidding strategies according to Algorithm~\ref{alg:mw_variant} (with any $\eta,\xi > 0$). Further, suppose the market allocation mechanism is Algorithm \ref{alg:mech}. There exists an integer $N(\alpha)$ such that for any $t \ge N ( \alpha )$, the number of shares cleared at time $t$ satisfies
\begin{align*}
	 \Pr \left[\Pi \left(p_t,\ts_t,\tb_t \right) \ge \opt - \frac{2 \ln (V/\alpha)}{\epsilon} - \frac{2 \ln \left( 1/\alpha \right)}{\epsilon}  - \sqrt{6  \left(\opt + \frac{\ln (1/\alpha)}{\epsilon} \right) \ln \left(1/\alpha\right)} \right] \ge 1 - 9\alpha.
\end{align*}
where this probability is taken with respect to the randomness of both Algorithms \ref{alg:mech} and \ref{alg:mw_variant}. 
\end{thm}

The proof idea is the following: despite the price randomness due to privacy, the event in which all agents bid their value and an optimal price is picked  happens infinitely often, as in the non-private case. In turn, (at least) $OPT$ buyers (resp. sellers) eventually learn to bid above (resp. below) an optimal price $p^\star$. However, the mechanism will still pick sub-optimal prices to guarantee privacy, as per the bound of Theorem~\ref{thm:thm_1}. We refer the reader to Appendix~\ref{app:dynamics_private} for a complete proof. 
A similar statement holds, with respect to benchmark $\opt'$ (see Definition \ref{def:opt'}), when agents update according to Algorithm~\ref{alg:mw}.

\begin{thm}
 Suppose buyers and sellers use the Exponential Weights Algorithm \ref{alg:mw} (with any $\eta > 0$) to update their bids. Further, suppose the market allocation mechanism is Algorithm \ref{alg:mech}. There exists an integer $N(\alpha) > 0$ such that for all $t \geq N(\alpha)$, the number of shares cleared at time $t$ satisfies
\begin{align*}
\Pr \left[ \Pi \left(p_t,\ts_t,\tb_t \right) \geq \opt' - \frac{2 \ln (V/\alpha)}{\epsilon} - \frac{2 \ln \left( 1/\alpha \right)}{\epsilon}  - \sqrt{6  \left(\opt' + \frac{\ln (1/\alpha)}{\epsilon} \right) \ln \left(1/\alpha\right)} \right] \ge 1 - 9\alpha.
\end{align*}
\end{thm}

%% file: appendix.tex

\input{dptools}


\section{Proofs of Privacy guarantees of our mechanisms}\label{app:DPproofs}

\begin{proof}[proof of Claim \ref{clm:DP_1}]
We start the proof by noticing that the sensitivity of $\Pi$ (as per Definition~\ref{def:exp}) is 1: indeed, changing one element in the data $(\vs,\vb)$, i.e. the valuation of a single agent, will change the number of shares cleared by at most one, for any fixed price $p$. We can therefore conclude that by Theorem \ref{thm:exp} the mechanism that takes the data set as input and outputs a price $p \propto \exp \left( \epsilon \Pi (p,\vs,\vb) / 2 \right)$ is $\epsilon$-DP. One can similarly argue that given a fixed price $p$, quantities $\sum_{i \in \S}  \mathbf{1} [p \geq \vs_i]$ and $\sum_{j \in \B} \mathbf{1} [p \leq \vb_j ]$ have sensitivity 1 (see Definition \ref{def:laplace}), and therefore by Theorem~\ref{thm:laplace}, $\widehat{s}$ and $\widehat{b}$ both satisfy $\epsilon$-DP. We can now invoke the Composition Theorem \ref{thm:composition} to conclude that the triplet $(p,\widehat{s}, \widehat{b})$ computed in Algorithm \ref{alg:mech} satisfies $3\epsilon$-differential privacy. The claim then follows by the Billboard Lemma \ref{lem:billboard} and noticing that each agent's allocation depends only on their own data and the triplet $(p,\widehat{s}, \widehat{b})$.
\end{proof}

\begin{proof}[Proof of Claim \ref{clm:DP_2}]
Notice first that according to Definition \ref{def:exp}, the sensitivity of $\Pi$ is 1 because for any $p$, changing one agent's valuation can change $\Pi$ by at most 1. Now fixing the price $p$ output by the first exponential mechanism, it similarly follows that the sensitivity of loss functions $L^s$ and $L^b$ are 2. We therefore have that the exponential mechanisms outputting $p$, $\tau^s$, and $\tau^b$ are all $\epsilon$-DP by Theorem \ref{thm:exp}, and hence the triplet $(p,\tau^s,\tau^b)$ satisfies $3\epsilon$-DP by the Composition Theorem \ref{thm:composition}. The claim then follows by the Billboard Lemma \ref{lem:billboard} and noticing that each agent's allocation depends only on their own data and the triplet $(p,\tau^s, \tau^b)$.
\end{proof}

\begin{proof}[Proof of Claim \ref{clm:DP_3}]
First, we note that the sensitivity of $f$ is upper-bounded by $\sqrt{6 \ln(1/\alpha)}$. Indeed, we remind the reader that $\opt$ has sensitivity of 1, and note that for any $x \ge 0$
\[
\sqrt{6 \ln(1/\alpha)} \geq \sqrt{ 6  \left(x + 1 + \frac{\ln (1/\alpha)}{\epsilon} \right) \ln \left(1/\alpha\right)} - \sqrt{ 6  \left(x + \frac{\ln (1/\alpha)}{\epsilon} \right) \ln \left(1/\alpha\right)},
\]
using the classical inequality $\sqrt{a} + \sqrt{b} \geq \sqrt{a+b}$. Therefore, by Theorem~\ref{thm:laplace}, the computation of $\widetilde{f}$ is $\epsilon$-differential private. In each of mechanisms $\M_1$ and $\M_2$, $p_1$ the price output by $\M_1$, respectively $p_2$ the price output by $\M_2$, are computed in an $\epsilon$-differentially private manner. Similarly, $\hat{s},~\hat{b}$ (resp. $\tau^s,~\tau^b$), the private counts of the number of agents willing to trade in $\M_1$ at price $p_1$ (resp. the thresholds picked by mechanism $\M_2$ for price $p_2$) are each the result of an $\epsilon$-differentially private conputation (conditional on $p_1,~p_2$). In turn, our mechanism can be seen as one that computes $(\widetilde{f},p_1,p_2,\hat{s},\hat{b},\tau^s,\tau^b)$ in a $7\epsilon$-differentially private manner (by the composition guarantee of Theorem~\ref{thm:composition}), then outputs an allocation $\os_i$ for each given seller $i$ (resp. $\ob_j$ for each buyer $j$) as a function of only $\vs_i$ (resp. $\vb_j$) and $(\widetilde{f},p_1,p_2,\hat{s},\hat{b},\tau^s,\tau^b )$. Hence, by Lemma~\ref{lem:billboard}, $\M$ is $7\epsilon$-joint differentially private.  
\end{proof}

\section{Proofs of Profit and Inventory of our Mechanisms}
\subsection{Proof of Theorem~\ref{thm:thm_1}}\label{app: mechanism1}
\begin{proof}
We will be using the following concentration inequalities in our proof.
\begin{fact}[Multiplicative Chernoff Bound]\label{fact:multchernoff}
Let $\{X_i\}_{i=1}^n$ be a collection of independent random variables where $ X_i \in [0,1]$ for all $i$. Let $S=\sum_{i=1}^n X_i$ and $\mu = \E [S]$. We have that for any $0 \le t \le 1$,
$$
\Pr \left[ S < (1-t) \mu \right] \le e^{-\frac{\mu t^2}{2}}
$$
\end{fact}

\begin{fact}[Bernstein's Inequality]\label{fact:bernstein}
Let $\{X_i\}_{i=1}^n$ be a collection of $i.i.d$ random variables where for each $i$, $ X_i \in [0,1]$, $E[X_i] = \mu$, and $Var(X_i) = \sigma^2$. Let $S=\sum_{i=1}^n X_i$. We have that for any $t \ge 0$,
$$
\Pr \left[ \left| S - n\mu \right| > t \right] \le 2 e^{-\frac{t^2}{2n\sigma^2 + 2t/3}}
$$ 
\end{fact}

Let $s(p) = \sum_{i \in S}  \mathbf{1}\left[p \geq \vs_i\right]$ and $b(p) = \sum_{j \in B} \mathbf{1} \left[p \leq \vb_j \right]$ be the number of sellers and buyers available at price $p$, where $p$ is the price chosen by the exponential mechanism. Note that as $p$ is a random variable, so are $s(p)$ and $b(p)$. From now on, for simplicity of notations, we omit the dependency of $s$ and $b$ in $p$. We start the proof by noting that by the accuracy guarantee of the Laplace mechanism (see Theorem \ref{thm:laplace}) and a union bound, we have with probability at least $1-2\alpha$ that 
\begin{align}\label{eq:Laplace_acc}
\left\vert \widehat{b} -  b\right\vert \leq \frac{\ln \left( 1/\alpha \right)}{\epsilon}, \quad \left\vert \widehat{s} -  s \right\vert \leq \frac{\ln \left( 1/\alpha \right)}{\epsilon},
\end{align}
and by the accuracy guarantee of the Exponential mechanism (see Theorem \ref{thm:exp}), we have with probability at least $1-\alpha$ that,
\begin{align}\label{eq:exp_acc}
\Pi \left( p, \vs, \vb \right) = \min \left\{ s, b \right\} \ge \opt - \frac{2 \ln (V/\alpha)}{\epsilon}.
\end{align}

By a union bound, Equations~(\ref{eq:Laplace_acc}) and~(\ref{eq:exp_acc}) hold simultaneously with probability at least $1 - 3\alpha$, and throughout this proof we condition on these events. Let $\widetilde{s} = \sum_{i \in \S} \os_i$ and $\widetilde{b} = \sum_{j \in \B} \ob_j$ be the number of sellers and buyers who participate in a trade, output by the mechanism.
First, let's focus on $\widetilde{s}$. Observe that
$$\widetilde{s} \, | s, \widehat{s}, \widehat{b} \sim Binomial \left(s, \widehat{q} = \min\left\{1, \frac{\left( \widehat{b} \right)_+ }{\left( \widehat{s} - \frac{\ln (1/\alpha)}{\epsilon}  \right)_+ }\right\} \right).
$$
Note that we have
\begin{align*}
\widehat{s} - \frac{\ln (1/\alpha)}{\epsilon} 
&\geq s - \frac{2\ln (1/\alpha)}{\epsilon} \tag{by Equation~\eqref{eq:Laplace_acc}}
\\&\geq \opt - \frac{2 \ln (1/\alpha)}{\epsilon} - \frac{2 \ln (V/\alpha)}{\epsilon} \tag{by Equation~\eqref{eq:exp_acc}}
\\&\geq \opt - \frac{4 \ln (V/\alpha)}{\epsilon} \tag{$V\ge1$}
\\&\geq \frac{\ln (V/\alpha)}{\epsilon} 
\tag{by assumption, $OPT \geq  \frac{5 \ln (V/\alpha)}{\epsilon}$}
\\&> 0,
\end{align*}
and also
\begin{align*}
\widehat{b}
&\geq b - \frac{\ln (1/\alpha)}{\epsilon} \tag{by Equation~\eqref{eq:Laplace_acc}}
\\&\geq \opt - \frac{ \ln (1/\alpha)}{\epsilon} - \frac{2 \ln (V/\alpha)}{\epsilon} \tag{by Equation~\eqref{eq:exp_acc}}
\\&\geq \opt - \frac{3 \ln (V/\alpha)}{\epsilon} \tag{$V\ge1$}
\\&\geq \frac{2 \ln (V/\alpha)}{\epsilon} 
\tag{by assumption, $OPT \geq  \frac{5 \ln (V/\alpha)}{\epsilon}$}
\\&> 0.
\end{align*}
As such, $\widehat{q}$ is well-defined, and we can rewrite it as
\begin{align*}
\widehat{q} = \min \left(1, \frac{\widehat{b}}{\widehat{s} - \frac{\ln (1/\alpha)}{\epsilon}}\right).
\end{align*}

We have by the multiplicative Chernoff Bound (Fact \ref{fact:multchernoff}) with $t= \sqrt{\frac{2 \ln (1/\alpha)}{s \widehat{q}}}$ that,
\begin{equation}\label{eq:Hoeffding}
 \widetilde{s} \ge s \widehat{q} - \sqrt{2s \widehat{q} \ln (1/\alpha)}
\end{equation}
with probability at least $1- \alpha$ when $t \leq 1$. Note that the bound applies when $t > 1$ too, noting that then $s \widehat{q} - \sqrt{2s \widehat{q} \ln (1/\alpha)} < 0$ but $\widetilde{s} \geq 0$. In what follows, we will provide an upper bound and a lower bound for the term $s \widehat{q}$ so that we can further lower bound $\widetilde{s}$ in Equation (\ref{eq:Hoeffding}). Symmetrically, we can get a similar lower bound for $\widetilde{b}$ which completes the first part of the proof because $\Pi (\M) = \min \{ \widetilde{s}, \widetilde{b} \}$.

On the one hand, note that
\begin{align}\label{eq:lower_bound_sp}
\begin{split}
s \widehat{q} 
&= s \cdot \frac{ \min \left\{ \widehat{s} - \frac{\ln (1/\alpha)}{\epsilon}  ,  \widehat{b} \right\}}{ \widehat{s} - \frac{\ln (1/\alpha)}{\epsilon}  } 
\\&\geq \frac{s}{s } \cdot  \min \left\{  \widehat{s} - \frac{\ln (1/\alpha)}{\epsilon}   ,  \widehat{b}  \right\}
\\& =  \min \left\{  \widehat{s} - \frac{\ln (1/\alpha)}{\epsilon}   ,  \widehat{b}  \right\}
\\&\geq \min \left\{ s, b \right\} - \frac{2 \ln \left( 1/\alpha \right)}{\epsilon}
\\&\geq \opt - \frac{2 \ln  \left(1/\alpha \right)}{\epsilon} -  \frac{2 \ln (V/\alpha)}{\epsilon}.
\end{split}
\end{align}
The first inequality follows from $\widehat{s} - \frac{\ln (1/\alpha)}{\epsilon}  \leq s$ by Equation~\eqref{eq:Laplace_acc}. The second inequality follows from
\[
\widehat{s} - \frac{\ln (1/\alpha)}{\epsilon}  \geq s - \frac{2 \ln \left( 1/\alpha \right)}{\epsilon}, \quad \widehat{b} \geq b - \frac{\ln (1/\alpha)}{\epsilon}
\] 
by Equation~\eqref{eq:Laplace_acc}. The third inequality is a direct application of Equation~\eqref{eq:exp_acc}. On the other hand, 
\begin{align}\label{eq:upper_bound_sp}
\begin{split}
s \widehat{q} &= s \cdot \frac{ \min \left\{  \widehat{s} - \frac{\ln (1/\alpha)}{\epsilon}   ,  \widehat{b}  \right\}}{ \widehat{s} - \frac{\ln (1/\alpha)}{\epsilon}  }
\\&\leq  \frac{s}{ s - \frac{2 \ln (1/\alpha)}{\epsilon}  } \cdot \min \left\{ s ,  b + \frac{\ln (1/\alpha)}{\epsilon}  \right\}
\\&\leq  \frac{s}{ s - \frac{2 \ln (1/\alpha)}{\epsilon}  } \cdot  \left(\min \left\{ s ,  b \right\} + \frac{\ln (1/\alpha)}{\epsilon} \right)
\\&\leq  \frac{s}{ s - \frac{2 \ln (V/\alpha)}{\epsilon}  } \cdot  \left(\min \left\{ s ,  b \right\} + \frac{\ln (1/\alpha)}{\epsilon} \right)
\\&\leq  3 \left(\min \left\{ s ,  b \right\} + \frac{\ln (1/\alpha)}{\epsilon} \right)
\\&\leq 3 \left( \opt + \frac{\ln (1/\alpha)}{\epsilon} \right).
\end{split}
\end{align}
The first inequality follows from 
\[
s - \frac{\ln (1/\alpha)}{\epsilon} \le \widehat{s} \leq s + \frac{\ln (1/\alpha)}{\epsilon}, \quad \widehat{b} \leq b + \frac{\ln (1/\alpha)}{\epsilon}
\]
by Equation~\eqref{eq:Laplace_acc}. The second-to-last inequality follows from the fact that
\[
f:(0,+\infty) \to \mathbb{R}, \quad f(x) =  \frac{x}{ x - \frac{2 \ln (V/\alpha)}{\epsilon} } = \frac{1}{ 1 - \frac{2 \ln (V/\alpha)}{\epsilon x} }
\]
is a non-increasing function of $x$, and that by Equation~\eqref{eq:exp_acc},
\[
s \geq \opt - \frac{2 \ln (V/\alpha)}{\epsilon} \geq \frac{3 \ln (V/\alpha)}{\epsilon}.
\]

Combining Equations~\eqref{eq:Hoeffding},~\eqref{eq:lower_bound_sp}, and~\eqref{eq:upper_bound_sp}, we obtain that with probability $1-4\alpha$,
\begin{align}\label{eq:lower_bound_s}
\begin{split}
&\widetilde{s} 
\geq  \opt - \frac{2 \ln \left( 1/\alpha \right)}{\epsilon} -  \frac{2 \ln (V/\alpha)}{\epsilon} - \sqrt{ 6 \left(\opt + \frac{\ln (1/\alpha)}{\epsilon} \right) \ln \left(1/\alpha\right)}.
\end{split}
\end{align}
Symmetrically, we can get the same bound for $\widetilde{b}$: with probability $1-4\alpha$,
\begin{align}\label{eq:lower_bound_b}
\begin{split}
&\widetilde{b} 
\geq  \opt - \frac{2 \ln \left( 1/\alpha \right)}{\epsilon} -  \frac{2 \ln (V/\alpha)}{\epsilon} - \sqrt{ 6 \left(\opt + \frac{\ln (1/\alpha)}{\epsilon} \right) \ln \left(1/\alpha\right)}
\end{split}
\end{align}
Combining Equations~\eqref{eq:lower_bound_s} and~\eqref{eq:lower_bound_b} and noting that $\Pi (\M_1) = \min \{ \widetilde{s}, \widetilde{b}\}$ proves the first part of the theorem. We conclude the proof by noting that the statements hold with probability at least $1-8\alpha$ by union bound.

Let us now analyze the inventory of the mechanism. We have by the triangle inequality that
\begin{equation}\label{eq:cost}
I \left(\M_1 \right) = \left| \widetilde{s} - \widetilde{b} \right| \le \left| \widetilde{s} - \min \left\{s,b \right\} \right| + \left| \widetilde{b} - \min \left\{s,b \right\} \right|
\end{equation}
We will provide an upper bound for the first term, and by symmetry, an upper bound on the second will follow immediately. We have that by the triangle inequality
\begin{equation}\label{eq:cost_s}
\left| \widetilde{s} - \min \left\{s,b \right\} \right| \le \left| \widetilde{s} - s \widehat{q} \right| + \left| s \widehat{q} - \min \left\{s,b \right\} \right|
\end{equation}
First
\begin{align}\label{eq:cost_s_part1}
\begin{split}
\left| \widetilde{s} - s \widehat{q} \right| &\le \sqrt{2s \widehat{q} (1- \widehat{q}) \ln (2/\alpha)} + \frac{2 \ln (2/\alpha)}{3} \\
&\le \sqrt{2s \widehat{q} \ln (2/\alpha)} + \frac{2 \ln (2/\alpha)}{3} \\
&\le \sqrt{6 \left(\opt + \frac{\ln (1/\alpha)}{\epsilon} \right) \ln (2/\alpha) } + \frac{2 \ln (2/\alpha)}{3}
\end{split}
\end{align}
where the first inequality holds with probability $1-\alpha$ and follows from the Bernstein's inequality (Fact \ref{fact:bernstein}) with $\sigma^2 = \widehat{q} (1-\widehat{q})$ and taking $t=\frac{\ln (2/\alpha)}{3} + \sqrt{\frac{\ln^2 (2/\alpha)}{9}+2n\sigma^2 \ln(2/\alpha)}$. Notice that $t \le \frac{2\ln (2/\alpha)}{3} + \sqrt{2n\sigma^2 \ln(2/\alpha)}$. The last inequality follows from the upper bound developed in Equation (\ref{eq:upper_bound_sp}). Second,
\begin{equation}\label{eq:cost_s_part2}
\left| s \widehat{q} - \min \left\{s,b \right\} \right| \le \frac{9 \ln (1/\alpha)}{\epsilon}.
\end{equation}
This is because as we showed in Equation (\ref{eq:lower_bound_sp}), 
\[
- \frac{2 \ln \left( 1/\alpha \right)}{\epsilon} \le s \widehat{q} - \min \left\{ s, b \right\},
\]
and as we showed in Equation (\ref{eq:upper_bound_sp}),
\begin{align*}
s \widehat{q} - \min \left\{ s, b \right\} 
&\leq \frac{s}{ s - \frac{2 \ln (1/\alpha)}{\epsilon}  } \cdot  \left(\min \left\{ s ,  b \right\} + \frac{\ln (1/\alpha)}{\epsilon} \right)
- \min \left\{ s, b \right\} 
\\& = \left( \frac{s}{ s - \frac{2 \ln (1/\alpha)}{\epsilon} } - 1 \right) \cdot \min \left\{ s ,  b \right\} + \frac{s}{ s - \frac{2 \ln (1/\alpha)}{\epsilon} } \cdot \frac{\ln (1/\alpha)}{\epsilon}\\
&\le \left( \frac{ {2 \ln (1/\alpha) / \epsilon} }{s - \frac{2 \ln (1/\alpha)}{\epsilon}  } \right) \cdot s + \frac{s}{ s - \frac{2 \ln (1/\alpha)}{\epsilon} } \cdot \frac{\ln (1/\alpha)}{\epsilon} \\
&= \frac{s}{s - \frac{2 \ln (1/\alpha)}{\epsilon}} \cdot \frac{3 \ln (1/\alpha)}{\epsilon} \\
&\le \frac{9 \ln (1/\alpha)}{\epsilon}
\end{align*}
Because we showed in Equation~(\ref{eq:upper_bound_sp}) that
$
\frac{s}{s - \frac{2 \ln (1/\alpha)}{\epsilon}  } \le 3 
$.
Putting together Equations (\ref{eq:cost_s}), (\ref{eq:cost_s_part1}), and (\ref{eq:cost_s_part2}) we get that with probability $1-3\alpha$,
\begin{align*}
\left| \widetilde{s} - \min \left\{s,b \right\} \right|
&\le \sqrt{6 \left(\opt + \frac{\ln (1/\alpha)}{\epsilon} \right) \ln (2/\alpha) } + \frac{2 \ln (2/\alpha)}{3} 
+ \frac{9 \ln(1/\alpha)}{\epsilon}.
\end{align*}
Swapping the roles of the buyers and a similar proof yields the same bound on $\left| \widetilde{b} - \min \left\{s,b \right\} \right|$.
A union bound completes the proof, by Equation (\ref{eq:cost}).
\end{proof}

\subsection{Proof of Theorem~\ref{thm:thm_2}}\label{app:mechanism2}
\begin{proof}
Let us define $\widetilde{s} = \sum_{i \in \S} \os_i$ and $\widetilde{b} = \sum_{j \in \B} \ob_j$ to be the number of sellers and buyers who participate in a trade under output allocation $\oo$. We have that with probability at least $1-3\alpha$, by the accuracy guarantee of the Exponential mechanism (see Theorem \ref{thm:exp}),
\begin{align}\label{eq:exp_acc1}
\Pi \left(p, \vs, \vb \right) \ge \opt - \frac{2 \ln (V/ \alpha)}{\epsilon},
\end{align}
and that since $\min_{\tau^s} L^s(\tau^s, p, \vs, \vb) = \min_{\tau_b} L^b(\tau^b, p, \vs, \vb) = 0$,
\begin{align}\label{eq:exp_acc2}
\left| \widetilde{s}  - \Pi(p, \vs, \vb) \right| = L^s(\tau^s, p, \vs, \vb) \le \frac{4 \ln \left(\ns/ \alpha \right)}{\epsilon} \quad \Longrightarrow \quad \widetilde{s} \ge \Pi(p, \vs, \vb) - \frac{4 \ln \left(\ns/ \alpha \right)}{\epsilon},
\end{align}
\begin{align}\label{eq:exp_acc3}
\left| \widetilde{b}  - \Pi(p, \vs, \vb) \right| = L^b(\tau^b, p, \vs, \vb) \le \frac{4 \ln \left(\nb/ \alpha \right)}{\epsilon} \quad \Longrightarrow \quad \widetilde{b} \ge \Pi(p, \vs, \vb) - \frac{4 \ln \left(\nb/ \alpha \right)}{\epsilon}
\end{align}
We therefore have that
\begin{align}
\Pi \left( \M_2 \right) = \min \left\{ \widetilde{s}, \widetilde{b} \right\} \ge \opt - \frac{2 \ln (V/ \alpha)}{\epsilon} - \frac{4 \ln (n/ \alpha)}{\epsilon}
\end{align}
Let's now analyze the inventory introduced by the private mechanism. We have that
\begin{align*}
I \left( \M_2 \right) &= \left| \widetilde{s} - \widetilde{b} \right| \\
&\le \left| \widetilde{s} - \Pi \left(p, \vs, \vb \right) \right| + \left| \widetilde{b} - \Pi \left(p, \vs, \vb \right) \right| \\
&= L^s(\tau^s, p, \vs, \vb) + L^b(\tau^b, p, \vs, \vb) \\
&\le \frac{4 \ln \left(\ns/ \alpha \right)}{\epsilon} + \frac{4 \ln \left(\nb/ \alpha \right)}{\epsilon} \\
&\le \frac{8 \ln (n/ \alpha)}{\epsilon}
\end{align*}
where the second inequality holds with probability $1-2\alpha$ by Equations (\ref{eq:exp_acc2}) and (\ref{eq:exp_acc3}).
\end{proof}

\subsection{Proof of Theorem~\ref{thm:thm_3}}\label{app:mechanism3}

\begin{proof}
This theorem follows from Theorems \ref{thm:thm_1} and \ref{thm:thm_2} and conditioning on the accuracy guarantee of the additional Laplace mechanism used in Algorithm \ref{alg:mech_combined}:
\begin{equation}\label{eq:lap_f}
\text{w.p. } 1 - \alpha, \quad \quad \left| \widetilde{f} - f \right| \le \frac{\sqrt{6} \ln^{1.5} (1/\alpha)}{\epsilon}
\end{equation}
Suppose $\widetilde{f} < 0$. Note that in this case,
\begin{align*}
\opt - \Pi \left( \M_3 \right) &= \opt - \Pi \left( \M_1 \right) \\
&\le \frac{2 \ln (V/\alpha)}{\epsilon} + \frac{2 \ln \left( 1/\alpha \right)}{\epsilon}  + \sqrt{6  \left(\opt + \frac{\ln (1/\alpha)}{\epsilon} \right) \ln \left(1/\alpha\right)} \tag{$\star$} \\
&= \frac{2 \ln (V/\alpha)}{\epsilon} + \frac{4 \ln (n/ \alpha)}{\epsilon} + f \\
&\le \frac{2 \ln (V/  \alpha)}{\epsilon} + \frac{4 \ln (n/ \alpha)}{\epsilon} + \widetilde{f} + \frac{\sqrt{6} \ln^{1.5} (1/\alpha)}{\epsilon} \\
&\le \frac{2 \ln (V/  \alpha)}{\epsilon} + \frac{4 \ln (n/ \alpha)}{\epsilon} + \frac{\sqrt{6} \ln^{1.5} (1/\alpha)}{\epsilon} \tag{$\star \star$}
\end{align*}
where the first inequality follows from Theorem \ref{thm:thm_1}, with probability $1-8\alpha$. The second inequality follows from Equation \ref{eq:lap_f}, with probability $1-\alpha$. Combining the bounds given by the second and the last inequalities (specified by $\star$ and $\star \star$), we get that with probability $1-9\alpha$,
\begin{align*}
\opt - \Pi \left( \M_3 \right) 
&\le  \min\left\{ \frac{2 \ln \left( 1/\alpha \right)}{\epsilon}  + \sqrt{ 6  \left(\opt + \frac{\ln (1/\alpha)}{\epsilon} \right) \ln \left(1/\alpha\right)},  \frac{4 \ln (n/ \alpha)}{\epsilon} \right\} 
\\&+ \frac{2 \ln (V/  \alpha)}{\epsilon} + \frac{\sqrt{6} \ln^{1.5} (1/\alpha)}{\epsilon}.
\end{align*}
A similar analysis for $\widetilde{f} \ge 0$ which uses Theorem \ref{thm:thm_2} gives us the same bound and proves the first part of the theorem. 

Now let's look at the inventory. Suppose $ \widetilde{f} \le 0$. We have that
\begin{align*}
    I \left( \M_3 \right) &= I \left( \M_1 \right) \\
    &\le \frac{18 \ln (1/\alpha)}{\epsilon} + 2\sqrt{6 \left(\opt + \frac{\ln (1/\alpha)}{\epsilon} \right) \ln (2/\alpha) } + \frac{4 \ln (2/\alpha)}{3} \\
    &\le \frac{18 \ln (1/\alpha)}{\epsilon} + 4\sqrt{6 \left(\opt + \frac{\ln (1/\alpha)}{\epsilon} \right) \ln (1/\alpha) } + \frac{4 \ln (2/\alpha)}{3} \tag{$\star$}\\
    &= 4 \left( f + \frac{4 \ln (n/ \alpha)}{\epsilon} \right) + \frac{10 \ln (1/\alpha)}{\epsilon} + \frac{4 \ln (2/\alpha)}{3} \\
    &\le 4 \left( \widetilde{f} + \frac{\sqrt{6} \ln^{1.5} (1/\alpha)}{\epsilon} + \frac{4 \ln (n/ \alpha)}{\epsilon} \right) + \frac{10 \ln (1/\alpha)}{\epsilon} + \frac{4 \ln (2/\alpha)}{3} \\
    &\le 4 \left(\frac{\sqrt{6} \ln^{1.5} (1/\alpha)}{\epsilon} + \frac{4 \ln (n/ \alpha)}{\epsilon} \right) + \frac{10 \ln (1/\alpha)}{\epsilon} + \frac{4 \ln (2/\alpha)}{3} \tag{$\star \star$}
\end{align*}
where the first inequality follows from Theorem \ref{thm:thm_1}, with probability $1-6\alpha$. The second inequality follows because $\alpha < 1/2$ (note we need $\alpha < 1/18$ to give non-trivial guarantee for the payoff of the mechanism). The third inequality follows from Equation \ref{eq:lap_f} with probability $1-\alpha$. Looking at the bounds given by the third and the last lines of the above equation (specified by $\star$ and $\star \star$), we get that with probability $1-7\alpha$,
\begin{align*}
I \left( \M_3 \right) 
&\le 4 \min\left\{ \frac{2 \ln \left( 1/\alpha \right)}{\epsilon}  + \sqrt{ 6  \left(\opt + \frac{\ln (1/\alpha)}{\epsilon} \right) \ln \left(1/\alpha\right)},  \frac{4 \ln (n/ \alpha)}{\epsilon} \right\} 
\\&+ \frac{4 \sqrt{6} \ln^{1.5} (1/\alpha)}{\epsilon} + \frac{10 \ln (1/\alpha)}{\epsilon} + \frac{4 \ln (2/\alpha)}{3}.
\end{align*}
A similar analysis for $\widetilde{f} \ge 0$ which uses Theorem \ref{thm:thm_2} gives us the same bound and proves the second part of the theorem.
\end{proof}

\section{Proof of Theorem~\ref{thm:lowerbounds}}\label{app:lowerbounds}

Consider the following family of data sets: first, we initialize $D_0$ as the data set that has $n$ sellers with valuations $\{1,\ldots,n\}$, and $n$ buyers with valuations $\{n,\ldots,2n-1\}$.
We then recursively construct $D_l$ for all $l$. To construct $D_{l+1}$ from $D_l$, we increase all valuations in $D_l$ by $1$, and assign buyers' (resp. sellers) identities in $D_{l+1}$ such that all buyers (resp. sellers) except one have the same valuation as in $D_l$. Equivalently, our construction works as follows: for any $l \in \mathbb{N}$,
$$
D_l = \begin{cases} \vs = \{l+1,\ldots,l+n\} & \text{sellers' valuations}\\
\vb = \{l+n,\ldots,l+2n-1\} & \text{buyers' valuations},
\end{cases}
$$
up to re-ordering of the agents' identities. The result will follow from the fact that a differentially private algorithm should output similar distributions of prices on data sets $D_0$ and $D_l$, but that at the same time, for $l$ large enough, $D_0$ and $D_l$ are far enough from each other that no distribution of prices can perform well over both of them. 

We first show the following lemma, which will be of use in the proof of Theorem~\ref{thm:lowerbounds}:
\begin{lemma}\label{lem:comp}
	 Let $\{D_l\}$ be the family of data sets described above. If $\A:~\D^n \to P$ is an $(\epsilon,\delta)$-DP algorithm, then for every price $p \in P$ and every $k,m \in \mathbb{N}$:
	 \begin{align*}
	 \Pr\left[\left|\A(D_k) - p\right| < m\right] \ge e^{-2k \epsilon} \Pr\left[\left|\A(D_0) - p\right| < m\right] - 2 k \delta.
	 \end{align*}
\end{lemma}
\begin{proof}
	By the definition of $(\epsilon,\delta)$-DP, if $D$ and $D'$ are neighboring data sets, we must have that for any event $E$,
	\begin{align*} 
	\Pr[\A(D) \in E] \leq e^{\epsilon} \Pr[ \A(D') \in E] + \delta,
	\end{align*}	
	or equivalently
	\begin{align}\label{eq:dp-inequality}
	\Pr[\A(D') \in E] \geq e^{-\epsilon}\left(\Pr[\A(D) \in E]-\delta\right)
	\end{align}
    Notice that for every $k$, by construction, $D_k$ and $D_{k+1}$ differ by only two entries (one buyer's and one seller's valuation). This immediately implies that $D_k$ and $D_0$ differ by at most $2k$ entries, hence we can apply inequality~\eqref{eq:dp-inequality} recursively $2k$ times to obtain that for any event $E$,
	\begin{align*}
	\Pr[\A(D_k) \in E] &\geq e^{-\epsilon} \left( e^{-\epsilon} \ldots  \left(e^{-\epsilon}\Pr[\A(D_0)\in E] - \delta \right) \ldots -\delta \right)-\delta\\ & = e^{-2k\epsilon}\Pr[\A(D_0)\in E] - \delta(e^{-(2k-1)\epsilon} +e^{-(2k-2)\epsilon} + .... + e^{-\epsilon} + 1)
	\\& \geq e^{-2k\epsilon} \Pr[\A(D_0) \in E] - 2 k \delta 
	\end{align*}
	where the last inequality follows from the fact that $e^x \le 1$ for $x \le 0$. Fixing the price $p$ and $k, m$, and taking $E$ to be the ball of radius $m$ around $p$, i.e.
	$$
	E= \{p': |p' - p| < m \}
	$$
	concludes the proof.
\end{proof}

We are now ready to prove Theorem \ref{thm:lowerbounds}. 
\begin{proof}[Proof of Theorem \ref{thm:lowerbounds}]
In this proof, for any given data set $D = (\vs,\vb)$, we let
\[
u(\A, D) 
\triangleq 
 \min \left\{\sum_{i \in \S} \mathbbm{1} \left[\vs_i \leq p \right], \sum_{j \in S} \mathbbm{1} \left[\vb_j \geq p \right] \right\}
\]
where $p$ is drawn according to $\A(D)$.

First, we note that in data set $D_0$, at most $n$ trades (where every trading agent gets non-negative utility) can occur, setting a price of $n$. Further, $n$ is the unique price that makes $n$ trades possible, noting that decreasing (resp. increasing) the price leads to strictly less than $n$ sellers (resp. buyers) willing to trade at that price. We let $p_0^* = n$ be this (unique) optimal price that clears $n$ shares on data set $D_0$. For a given $(\epsilon,\delta)$-DP algorithm $\A:~\D^n \to P$ that outputs a price $p$ given an input data set $D = (\vs,\vb)$, let us define, for any $k, m \in \mathbb{N}$ (we will choose these values later on),
	\begin{align*}
	q_m^0 := \Pr\left[|\A(D_0) -p_0^*| < m\right],~q_m^k := \Pr\left[|\A(D_k) -p_0^*| < m\right].
	\end{align*}
Notice by Lemma \ref{lem:comp} that
\begin{align}\label{eq:probbound}
q_m^k \geq e^{-2k\epsilon} q_m^0 - 2 k \delta.
\end{align}

Now, fix $m = \lceil \frac{1}{\epsilon} \rceil$, $k = 2\lceil\frac{1}{\epsilon}\rceil$, and take $n \geq m$. We have that the expected loss of $\A$ on $D_0$ is
\begin{align}\label{eq:D0-reg}
\begin{split}
\E_{\A} \left[ L(\A, D_0) \right] &= \opt(D_0) - \E_{\A} \left[ u(\A, D_0) \right] \\
&= n - \E_{\A} \left[ u(\A, D_0) \right] \\
&\ge n - \left( q_m^0 \cdot n+ (1- q_m^0) \cdot (n-m) \right) \\
&= (1- q_m^0) \cdot m \\
&\ge (1- q_m^0) \cdot \left(\frac{1}{\epsilon} \right).
\end{split}
\end{align}
The first inequality follows from a simple application of the law of total expectation on event $E= \{p: |p - p_0^*| < \frac{m}{2n} \}$ and its complement: with probability $1-q_m^0$ the outputted price is outside $E$, which implies that it can only clear at most $n-m \geq 0$ shares (picking a price that is $m$ away from $n$ necessarily leads to either $m$ fewer buyers or $m$ fewer sellers willing to trade); the rest of the time, with probability $q_m^0$, algorithm $\A$ clears at most $n$ shares. The second inequality is an immediate consequence of the choice of $m$. Similarly, on data set $D_k$,
\begin{align}\label{eq:Dk-reg}
\begin{split}
\E_{\A} \left[ L(\A, D_k) \right] &= \opt_k - u(\A, D_k) \\
&= n - \E_{\A} \left[ u(\A, D_k) \right] \\
&\ge n - \left( q_m^k \cdot (n - (k-m)) + (1- q_m^k) \cdot n \right) \\
&= q_m^k \cdot (k-m) \\
&\ge \left( e^{-2k\epsilon} q_m^0 - 2 k \delta \right) \cdot (k-m) \\
&\ge \left( e^{-8} q_m^0 - 8 (\delta / \epsilon) \right) \cdot \left(\frac{1}{\epsilon} \right)
\end{split}
\end{align}
where the first inequality follows from another use of the law of total expectation on the event $E$ and its complement (notice we choose our parameters so that $k > m$ and $n \geq k - m$): with probability $q_m^k$, the price is at most $n+m$, and there are $k-m$ sellers that are willing to trade at price $n+m$ but not at price $n+k$, implying that such a price clears at most $n- (k-m)$ shares; the rest of the time, the number of shares cleared is at most $n$ always. The second follows from Equation (\ref{eq:probbound}) and the last one follows from the choice of $k$ and $m$ and the fact that $\epsilon \lceil \frac{1}{\epsilon} \rceil \le 1 + \epsilon \le 2$ for $0 \le \epsilon \le 1$. Now let $L(\A)$ be the worst-case expected loss of $\A$. We have that
\begin{align*}
L(\A) &\ge \max \left\{(1- q_m^0),   \left( e^{-8} q_m^0 - 8 (\delta / \epsilon) \right) \right\} \cdot \left(\frac{1}{\epsilon} \right) \\
&\ge \left( \frac{e^{-8} - 8 (\delta/\epsilon)}{1 + e^{-8}} \right) \cdot \left(\frac{1}{\epsilon} \right)
\end{align*}
where the first inequality follows from Equations (\ref{eq:D0-reg}) and (\ref{eq:Dk-reg}) and the second is a simple observation that $f(q_m^0) := \max \left\{(1- q_m^0),   \left( e^{-8} q_m^0 - 8 (\delta / \epsilon) \right) \right\}$ is minimized at $q_m^0 = \frac{1 + 8 (\delta / \epsilon)}{1 + e^{-8}}$. Notice the lower bound is valid only when $\delta < \frac{e^{-8}}{8}\epsilon = \bigo (\epsilon)$. This proves our claim that $L(\A) = \Omega \left( \frac{1}{\epsilon} \right)$.
\end{proof}

\section{Proofs of Approximate Truthfulness}\label{app:IC-IR}

Our proof of truthfulness for Mechanism~\ref{alg:mech} will leverage the following lemma, which shows the output of an $(\epsilon,0)$-DP mechanism does not change by much in expectation when the input data set is changed by at most one element.

\begin{lemma}  \label{lem:distance}
	Let $Y=\M(D)$ where $\M:D \to \mathcal{Y}$ is an $(\epsilon,0)$-DP mechanism, and let $\max_{y \in \mathcal{Y}} |y| \leq K$. Then for any neighboring data sets $D \sim D'$,
	\begin{align*}
 |	\E \left[ Y(D) \right] - \E \left[ Y (D') \right] |  \leq (e^\epsilon - 1) K
	\end{align*}
\end{lemma}
\begin{proof}
$Y(D)$ and $Y(D')$ are random variables; we represent the possible values they may take on as $y\in\mathcal{Y}$, and represent the probability distribution of $Y$ under $D$, $D'$ as $\mathcal{P}$, $\mathcal{P}'$, respectively. It follows that
	\begin{align*}
	\E \left[ Y(D) \right] - \E \left[ Y (D') \right] = \E_{Y \sim \mathcal{P}} Y - \E_{Y \sim \mathcal{P}'} Y = \sum_{y \in \mathcal{Y}} \left( \Pr_{\mathcal{P}}[Y=y] - \Pr_{\mathcal{P}'} [Y=y] \right) y
	\end{align*}
Therefore,
	\begin{align*}
	\left\vert\E \left[ Y(D) \right] - \E \left[ Y (D') \right] \right\vert 
	&\leq \sum_{y \in \mathcal{Y}} 
    \left\vert\Pr_{\mathcal{P}}[Y=y] - \Pr_{\mathcal{P}'}[Y=y]\right\vert \left\vert y \right\vert
    \\&\leq  
    \sum_{y \in \mathcal{Y}} (e^\epsilon-1) \max\left\{\Pr_{\mathcal{P}}[Y=y],~\Pr_{\mathcal{P}'}[Y=y]\right\} \left\vert y \right\vert 
    \\&\leq (e^\epsilon - 1) K,
	\end{align*}
	where the second inequality follows from the definition of $(\epsilon,0)$-differential privacy.
\end{proof}

\begin{proof}[Proof of Claim \ref{clm:appxic}]
	We prove the claim for any seller. A similar proof holds for buyers. Fix an index $i$, and any reports/bid vector $\left(\ts_{-i},\tb\right)$ for the remaining buyers and sellers. For simplicity of notation, let us denote $(\vs_i, \ts_{-i}, \tb)$ where $i$ submits his bid truthfully as data set $D$, and $(\ts_i, \ts_{-i}, \tb)$ for some (other) report $
	\ts_i$ as data set $D'$. Notice $D$ and $D'$ are neighboring data sets. Writing $\E_{\M}$ for the expectation with respect to the mechanism $\M$, we have that: 
\begin{align*}
\E_\M \left[\us_i (\M (D')) \right] &= \E_\M [ \os_i \cdot (p - \vs_i ) | D' ] \\
&= \E_\M \left[  \1 \left[ p \ge \ts_i \right] Bern(q^s) (p - \vs_i ) | D' \right] \\
&= \E_\M \left[  \1 \left[ p \ge \vs_i \right] \cdot \1 \left[ p \ge \ts_i \right] Bern(q^s)  (p - \vs_i ) | D' \right] \\
&+ \E_\M \left[  \1 \left[ p < \vs_i \right] \cdot \1 \left[ p \ge \ts_i \right] Bern(q^s)   (p - \vs_i ) | D' \right] \\
&\le \E_\M \left[  \1 \left[ p \ge \vs_i \right]  Bern(q^s) (p - \vs_i ) | D' \right] \\
&\le \E_\M \left[  \1 \left[ p \ge \vs_i \right] Bern(q^s)   (p - \vs_i ) | D \right] + \left(e^{3\epsilon} -1\right)V \\
&= \E_\M \left[\us_i (\M (D)) \right] + \left(e^{3\epsilon} -1\right)V
\end{align*}
where the first inequality follows because the second term appearing in the sum is nonpositive and that $\1 \left[ p \ge \vs_i \right] \cdot \1 \left[ p \ge \ts_i \right] \le \1 \left[ p \ge \vs_i \right]$. The second inequality follows from Lemma \ref{lem:distance} and the fact that the computation of the pair of random variables $(p,q^s)$ combined with any post-processing of the pair $(p,q^s)$ that is independent of the reported data $D' = (\ts,\tb)$ satisfies $(3\epsilon,0)$-differential privacy by the Post-processing Lemma \ref{lem:postproc} and the Composition Theorem \ref{thm:composition}. Also note that the price/bids range is $\{1,2, \ldots, V\}$, so we can take $K=V$ in Lemma \ref{lem:distance}.
\end{proof}

The proofs of approximate truthfulness of Mechanisms~\ref{alg:mech_lottery} and~\ref{alg:mech_combined} follow the exact same argument that leverages the stability properties of differential privacy. The only difference comes in the choice of tie-breaking rule and the level of differential privacy of Mechanisms~\ref{alg:mech_lottery} and~\ref{alg:mech_combined}. Rewriting the above proofs with the corresponding tie-breaking rules yields the argument.

%% file: dptools.tex
\section{Differential Privacy Tools}\label{app:dp_tools}

In this section, we remind the reader of mechanisms that are classically used to guarantee differential privacy. These mechanisms work by adding appropriately-chosen noise to the choices and outputs of a mechanism, so as to ensure that a change in a single individual's data cannot have a large distributional effect on the mechanism's output. The noise introduced by differentially private mechanisms depends not only on the level $(\epsilon,\delta)$ of privacy one aims to guarantee, but also on the \emph{sensitivity} of the query of interest. This sensitivity measures how much the real-valued function of interest is affected by a change in a single entry of an input data set, and will be formally defined in our introduced DP mechanisms.

A commonly used mechanism for releasing the answer to numerical queries while guaranteeing $(\epsilon,0)$-differential privacy is the Laplace mechanism. The Laplace mechanism takes a numerical query $f$ as an input, and perturbs the value of $f$ on the input data set with zero-mean Laplace noise that has scale proportional to $(\Delta f / \epsilon)$ where $\Delta f$ is the $\ell_1$-sensitivity of $f$.

\begin{definition}[Laplace Mechanism \cite{dwork}]\label{def:laplace}
Given a function $f: \D^n \to \Real^k$ with $\ell_1$-sensitivity $\Delta f$:
\[
\Delta f = \max_{\overset{D,D' \, \in \, \mathcal{D}^n}{D \sim D'}} \left\Vert f(D) - f(D') \right\Vert_1,
\]
a data set $D \in \D^n$, and a privacy parameter $\epsilon$, the Laplace mechanism outputs:
\[
f_{\epsilon} \left(D \right) = f \left(D \right) + \left(W_1, \ldots, W_k \right)
\]
where $W_i$'s are $i.i.d.$ random variables drawn from $\text{Lap} \left( \Delta f/ \epsilon \right)$.
\end{definition}

We provide the privacy and accuracy guarantees of the Laplace mechanism below:
\medskip
\begin{thm}[Privacy vs. Accuracy of the Laplace Mechanism \cite{dwork}]\label{thm:laplace}
The Laplace Mechanism guarantees $(\epsilon,0)$-differential privacy and that with probability at least $1-\delta$,
$$
\left\Vert f_\epsilon \left(D \right) - f \left(D \right) \right\Vert_\infty \le \ln \left(\frac{k}{\delta} \right) \cdot \left( \frac{\Delta f}{\epsilon} \right)
$$
\end{thm}

We remark that the Laplace mechanism can be used to privately output the answer to numerical queries. However, suppose we want to privately output the solution to a maximization problem defined on the input data. Then, directly adding noise to the optimal solution could completely destroy the objective value of the maximization problem in question (for example, in an auction, adding a small amount of noise on the price of an item could significantly reduce revenue). In such situations, the Laplace mechanism performs poorly, and a better choice of private mechanism is the Exponential Mechanism, defined below:

\begin{definition}[Exponential Mechanism \cite{McSherryT07}]\label{def:exp}
Let $U: \D^n \times P \to \Real$ be a utility function that takes a data set $D \in \D^n$ and a parameter $p \in P$ as inputs, and let $\Delta U$ be its sensitivity. In other words,
$$
\Delta U = \max_{p \,\in \, P} \max_{\overset{D, D' \, \in \, \mathcal{D}^n}{D \sim D'}} \left\vert U \left(D,p \right) - U \left(D',p \right) \right\vert.
$$
Given a data set $D \in \D^n$ and a privacy parameter $\epsilon$, the exponential mechanism outputs $p \in P$ with probability proportional to $\exp \left(\frac{\epsilon U(D,p)}{2 \Delta U} \right)$ where $\exp(\cdot)$ is the exponential function.
\end{definition}

\begin{thm}[Privacy vs. Accuracy of the Exponential Mechanism \cite{McSherryT07}]\label{thm:exp}
The Exponential Mechanism guarantees $(\epsilon,0)$-differential privacy. Further, let $p_\epsilon \in P$ be the output of the Exponential mechanism, we have that with probability at least $1-\delta$,
$$
\left\vert U \left(D,p_\epsilon \right) - \max_{p \in P} U \left(D, p \right) \right\vert \le \ln \left(\frac{|P|}{\delta} \right) \cdot \left( \frac{2 \Delta U}{\epsilon} \right)
$$
\end{thm}

An important property of differential privacy is that it is robust to \textit{post-processing}. Applying any data-independent function to the output of an $(\epsilon,\delta)$-DP algorithm preserves $(\epsilon,\delta)$-differential privacy.
\medskip
\begin{lemma}[Post-Processing \cite{dwork}]\label{lem:postproc}
Let $\M: \D^n \to \R$ be an $(\epsilon, \delta)$-DP algorithm and let $f: \R \to \R'$ be any function. We have that the algorithm $f \, o \, \M: \D^n \to \R'$ is $(\epsilon, \delta)$-DP.
\end{lemma}

Another important property of differential privacy is that DP algorithms can be composed adaptively with a graceful degradation in their privacy parameters.
\medskip
\begin{thm}[(Simple) Composition \cite{DworkRV10}]\label{thm:composition}
Let $\M_t$ be an $(\epsilon_t, \delta_t)$-DP algorithm for $t \in [T]$. We have that the composition $\M = (\M_1, \M_2, \ldots, \M_T)$ is $(\epsilon, \delta)$-DP where $\epsilon = \sum_{t} \epsilon_t$ and $\delta = \sum_t \delta_t$.
\end{thm}

To prove that our mechanisms satisfy $(\epsilon,\delta)$-joint differential privacy, we will leverage the billboard lemma (\cite{billboard}). The billboard lemma shows that for every individual $i$ in the data set, restricting $i$'s output to be a function of only the output of a differentially private mechanism (run on all agents' data) and his own input guarantees joint differential privacy. 
\begin{lemma}[Billboard Lemma \cite{billboard}]\label{lem:billboard}
Suppose $\mathcal{M}: \D^n \to \mathcal{R}'$ is $(\epsilon,\delta)$-differentially private. Consider any set of functions $f_i: \mathcal{D}_i \times \R' \to \R$, where $\mathcal{D}_i$ is the portion of the data set containing $i$'s data. The composition $\left\{ f_i( \Pi_i (D), \mathcal{M}(D) ) \right\}$ is $(\epsilon,\delta)$-jointly differentially private, where $\Pi_i: \D^n \to \D_i$ is the projection to $i$'s data.
\end{lemma}

%% file: appendix_dynamics.tex
\section{Proofs for Learning Dynamics}\label{app:dynamics}

\subsection{Proof of No-Regret Lemma~\ref{cor:noregret_variant}}\label{app:noregret_variant}
We first show the claim below:
\begin{clm}\label{lem:mw_var_noregret}
Let $R_{j,t}$ be the random variable representing the reward of buyer $j$ in Algorithm \ref{alg:mw_variant} at round $t$, and let $R^*_j (T)$ be the total reward of buyer $j$'s best fixed action in hindsight, over $T$ rounds. Moreover, let $\xi \leq V$
and $\eta \leq \frac{1}{V}$.  
Then, the regret of buyer $j$ over $T$ rounds is bounded as follows:
\begin{align} \label{eqn:noregret}
     R^*_j \left(T \right) - \E \left[ \sum_{t=1}^T R_{j,t} \right] \le \xi T + \eta V^2 T + \frac{\ln V}{\eta}
\end{align}
\end{clm}

\begin{proof}
We can think of Algorithm \ref{alg:mw_variant} as Exponential Weights with a modified utility function:
\begin{align*}
    \text{buyer $j$'s modified utility at time $t$ for bid $k$}: \quad \umodb_{j,t}(k) = \begin{cases} \xi \cdot q_t^b & k=\vb_j \text{ and } p_t = \vb_j \\
    u_{j,t}(k) & \text{ otherwise }
    \end{cases}
\end{align*}
where $u_{j,t}(k)$ is the actual utility of buyer $j$ at time $t$ if he were to bid $k$. Importantly, we show that using this modified utility function we can still achieve vanishing regret (with respect to the \emph{original} reward $R_{j,t}$ which is the agent's \emph{true/realized} utility).

First, notice that $u_{j,t}$ is always upper-bounded by $\umodb_{j,t}$: $u_{j,t} \le \umodb_{j,t}$; but also that $\umodb_{j,t} \leq u_{j,t} + \xi$. Recall $R^*_j (T)$ is the reward of the best fixed action in hindsight, with respect to the sequence of prices $p_1,\ldots,p_T$ and probabilities $q_1^b, \ldots, q_T^b$ as chosen by an adversary. Let $r^*_j$ be the report that leads to achieving $R^*_j$, i.e.
\begin{align*}
    R^*_j (T) \triangleq \max_{k \in \{1,...,V\}} \sum_{t=1}^{T} u_{j,t}(k),
    \quad r^*_j \triangleq \argmax_{k \in \{1,...,V\}} \sum_{t=1}^{T} u_{j,t}(k)
\end{align*}
Our goal will be to show that Equation (\ref{eqn:noregret}) holds. Our proof technique will mostly follow standard arguments. In this proof -- and this proof only -- we let $w$ denote the unnormalized weights that may not sum to $1$, and note they induce probability distributions $\rho$ by normalizing each weight by the sum of the weights.
First, let $W_{j,t} = \sum_{k=1}^{V} w_{j,t}(k)$. By definition: 
\begin{align*}
 \frac{W_{j,t+1}}{W_{j,t}} = \frac {\sum_{k=1}^{V} {w_{j,t+1}(k)}}{\sum_{k=1}^{V} w_{j,t}(k)} = \sum_{k=1}^{V} \frac{w_{j,t}(k)e^{\eta \umodb_{j,t}(k)}}{\sum_{k=1}^{V} w_{j,t}(k)}
\end{align*}
We will write $\rho_{j,t}(k)\triangleq w_{j,t}(k) / \sum_{k=1}^{V} w_{j,t}(k)$ as the probability distribution induced by weights $w_{j,t}(k)$, for all $k$. We can rewrite the above as 
\begin{align*}
    \frac{W_{j,t+1}}{W_{j,t}} = \sum_{k=1}^{V} \rho_{j,t}(k) e^{\eta \umodb_{j,t}(k)}.
\end{align*}
For $\eta \leq \frac{1}{V}$ and $\xi \leq 1$, we have $\eta \mu_{j,t}(k) \leq 1$ for all $k$. Using the upper bound that $e^{x} \leq 1+ x + x^2$ for all $x \in [0,1]$, we obtain that
\begin{align*}
    \frac{W_{j,t+1}}{W_{j,t}} \leq 1 + \sum_{k=1}^{V} \rho_{j,t} (k) \cdot \eta \umodb_{j,t}(k) + \sum_{k=1}^{V} \rho_{j,t} (k) \cdot \eta^2 \umodb_{j,t}(k)^2
\end{align*}
Then
\begin{align}\label{ineq:logweights}
\begin{split}
 \ln \frac{W_{j,t+1}}{W_{j,t}} 
 &\leq \ln \left(1+ \eta \sum_{k=1}^{V} \rho_{j,t}(k) \umodb_{j,t}(k) + \eta^2 \sum_{k=1}^{V} \rho_{j,t}(k) \umodb_{j,t}(k)^2 \right) 
 \\&\leq \eta \sum_{k=1}^{V} \rho_{j,t} (k) \umodb_{j,t}(k) + \eta^2 \sum_{k=1}^{V} \rho_{j,t} (k) \umodb_{j,t}(k)^2,
 \end{split}
\end{align}
where we have used the fact that $\ln \left(1+x \right) \leq x$ for $x >-1$ (which holds in this case because payoffs are nonnegative given buyers (sellers) never bid above (below) their valuations). 
Now noting that $\frac{W_{j,t+1}}{W_{j,1}} = \frac{W_{j,t+1}}{W_{j,t}} \frac{W_{j,t}}{W_{t-1}}\dots \frac{W_{j,2}}{W_{j,1}}$, we can express \begin{align*}
    \ln \frac{W_{j,t+1}}{W_{j,1}} 
    = \ln \frac{W_{j,t+1}}{W_{j,t}}\dots \frac{W_{j,2}}{W_{j,1}} 
    = \sum_{\tau=1}^{t} \ln \frac{W_{\tau+1,j}}{W_{\tau,j}}
\end{align*}
And applying Inequality \eqref{ineq:logweights}, we have that
\begin{align}\label{ineq:lwtime}
    \ln  \frac{W_{j,t+1}}{W_{j,1}} \leq \eta \sum_{\tau=1}^{t}\sum_{k=1}^{V} \rho_{j,\tau} (k) \umodb_{j,\tau}(k) + \eta^2 \sum_{\tau = 1}^{t} \sum_{k=1}^{V} \rho_{j,\tau}(k) \umodb_{j,\tau}(k)^2
\end{align}

On the other hand, since $W_{j,t+1} \geq w_{j,t}(k)$ for all $k$, \emph{including} for the best action in hindsight $k=r^*_j$, we have that
\begin{align*}
    \ln \frac{W_{j,t+1}}{W_{j,1}} \geq \ln \frac{w_{j,t+1}(r^*_j)}{W_{j,1}} &= \ln( e^{\eta \umodb_{j,t}(r^*_j)} w_{j,t}(r^*_j)/W_{j,1}) \\ &= \ln(e^{\eta \umodb_{j,t}(r^*_j)} e^{\eta \umodb_{j,t-1}(r^*_j)} w_{j,t-1}(r^*_j)) - \ln W_{j,1} \\&=.... \\ &= \ln \left(\prod_{\tau=1}^{t} e^{\eta \umodb_{j,\tau}(r^*_j)} w_{j,1}(r^*_j)\right) - \ln W_{j,1}.
\end{align*}
Now, using the fact that the weights can be initialized with $w_{j,1}(k) = 1 \ \forall k$ and $W_{j,1} = V$, this gives 
\begin{align}\label{ineq:rstar}
     \ln \frac{W_{j,t+1}}{W_{j,1}}  \geq \eta \sum_{\tau=1}^{t} \umodb_{j,\tau}(r^*_j) - \ln V
\end{align}
But now combining Inequalities (\ref{ineq:lwtime}) and (\ref{ineq:rstar}) gives:
\begin{align*}
    \eta \sum_{\tau=1}^{t} \umodb_{j,\tau}(r^*_j) - \ln V \leq \eta \sum_{\tau=1}^{t} \sum_{k = 1}^{V} \rho_{j,\tau}(k) \umodb_{j,\tau}(k) + \eta^2 \sum_{\tau=1}^{t}\sum_{k=1}^{V} \rho_{j,\tau}(k) \umodb_{j,\tau}(k)^2
\end{align*}
Now notice that $\sum_{k=1}^{V} \rho_{j,\tau}(k) \umodb_{j,\tau}(k) = \E_k [ \umodb_{j,\tau}(k)]$. So rearranging and letting $t=T$, we have that 
\begin{align*}
    \sum_{\tau=1}^{T} \umodb_{j,\tau}(r^*_j) - \sum_{\tau=1}^{T} \E_k[\umodb_{j,\tau}(k)] \leq \frac{\ln V}{\eta} + \eta \sum_{\tau=1}^{T} \E_k[\umodb_{j,\tau}(k)^2] \leq \frac{\ln V}{\eta} + \eta T V^2 
\end{align*}
where the inequality follows from the fact that $\umodb_{j,t}$ is bounded by $\max(V,\xi) = V$ (remembering that $u_{j,t} \leq V$). But since $\umodb_{j,\tau}(k) \geq u_{j,\tau}(k)$, we have that 
\begin{align*}
    R^*_j (T) - \sum_{\tau=1}^{T} \E_k[\umodb_{j,\tau}(k)]
     &=\sum_{\tau=1}^{T} u_{j,\tau}(r^*_j) - \sum_{\tau=1}^{T} \E_k[\umodb_{j,\tau}(k)] \\ 
     &\leq \sum_{\tau=1}^{T} \umodb_{j,\tau}(r^*_j) - \sum_{\tau=1}^{T} \E_k[\umodb_{j,\tau}(k)] \\
     &\leq \frac{\ln V}{\eta} + \eta T V^2
\end{align*}
Further, since $\umodb_{j,\tau}(k) \leq u^b_{j,t}(k) + \xi$, we also have that 
\begin{align*}
    R^*_j (T) - \sum_{\tau=1}^{T} \E[\umodb_{j,\tau}(k_{j,\tau})] &= \sum_{\tau=1}^{T} u_{j,\tau}(r^*_j) - \sum_{\tau=1}^{T} \E_k[\umodb_{j,\tau}(k)] \\
    &\ge \sum_{\tau=1}^{T} \umodb_{j,\tau}(r^*_j) - \sum_{\tau=1}^{T} \E_k[\umodb_{j,\tau}(k)] - \xi T\\
    &\ge \sum_{\tau=1}^{T} u_{j,\tau}(r^*_j) - \sum_{\tau=1}^{T} \E_k[u_{j,\tau}(k)] - \xi T \\
    &= R^*_j (T) - \sum_{t=1}^{T} \E[R_{j,t}] - \xi T. 
\end{align*}
Combining the last two inequalities, we get
\begin{align*}
    R^*_j \left(T \right) - \sum_{t=1}^{T} \E \left[R_{j,t} \right] \leq \xi T + \frac{\ln V}{\eta} + \eta V^2 T,
\end{align*}
as desired.
\end{proof}
We can now conclude the proof, noting that Lemma \ref{lem:mw_var_noregret} gives that the total regret of Algorithm \ref{alg:mw_variant} over $T$ rounds for agent $j$ is bounded by:
\begin{align*}
    \text{Regret} \leq  \xi T + \frac{\ln V}{\eta} + \eta V^2 T
\end{align*}
Choose $\eta = \frac{1}{V \sqrt{T}}$ 
and $\xi = \frac{1}{\sqrt{T}}$. Then we have that 
\begin{align*}
    \text{Regret} \leq \sqrt{T} + V \ln V \sqrt{T}  +  \frac{1}{ V\sqrt{T}} V^2 T = \sqrt{T} + V \ln V \sqrt{T} +  V \sqrt{T}. 
\end{align*}
Then average regret can be bounded as:
\begin{align*}
    \frac{1}{T} \text{ Regret} \leq \frac{1}{\sqrt{T}} + \frac{V \ln V}{\sqrt{T}} + \frac{V}{\sqrt{T}} = \bigo\left(\frac{1}{\sqrt{T}}\right).
\end{align*}
That is, average regret vanishes as $T \to \infty$.

\subsection{Proof of Theorem~\ref{thm:limit_1}}\label{app:dynamics_indifferent}
To prove Theorem~\ref{thm:limit_1}, we will examine how the $\opt'$ sellers with the lowest values and the $\opt'$ buyers with the highest values update their weights. To do so, we will need the following definition:
\begin{definition}[Highest (resp. lowest) value buyers (resp. sellers)]
Let $n^b(v) = \sum_{i=1}^{n^b} \1 [\vb_j \geq v ]$ be the number of buyers  with value bigger than or equal to $v$, and let $\nu^b = \max \{v:~ n^b(v) \geq \opt'\}$. Similarly, let $n^s(v) = \sum_{i=1}^{n^s} \1 [ \vb_j \leq v ]$ be the number of sellers with value smaller than or equal to $v$, and let $\nu^s = \min \{v:~ n^s(v) \geq \opt'\}$.
\end{definition}

We note the following property of $\nu^b, \nu^s$:
\begin{clm}\label{clm: buyer-seller gap}
Suppose $\opt' > 0$. Then,
\[
\nu^b \geq \nu^s + 2.
\]
\end{clm}

\begin{proof}
By definition of $\opt'$, there exists a price $p^\star$ such that at least $\opt'$ buyers have value above or equal to $p^\star+1$ and $\opt'$ sellers below or equal to $p^\star-1$. But then, $\nu^s \leq p^\star - 1$ and $\nu^b \geq p^\star +1$, which concludes the proof.
\end{proof}

First of all, we show that if a given price $p$ is picked infinitely many times, every agent $j$ with $\vb_j > p$ sees their probability of bidding more than $p$ converge to $1$. This is the object of Corollary~\ref{cor:limto1}, whose proof relies on Lemmas~\ref{lem: non-decreasing} and~\ref{lem:mass} below. We state the Lemmas for a buyer $j$ and note that similar results hold for a seller $i$ as well.

\begin{lemma}\label{lem: non-decreasing}
For all $t$, for all $p \in [V]$, for all $j \in [\nb]$, 
\[
\sum_{k = p}^{\vb_j} w^b_{j,t+1}(k) \geq \sum_{k = p}^{\vb_j} w^b_{j,t}(k).
\]
\end{lemma}

\begin{proof}
If $\sum_{k < p} w^b_{j,t}(k) = 0$, the result is immediate: it must be that for all $k < p$, $w^b_{j,t}(k) = 0$, so by exponential update, $w^b_{j,t+1}(k) = 0$, leading to $\sum_{k < p} w^b_{i,t + 1}(k) = 0$. In turn, 
\[
\sum_{k \geq p} w^b_{j,t+1}(k) = \sum_{k \geq p} w^b_{j,t}(k) = 1.
\]
We now focus on the case when $\sum_{k < p} w^b_{j,t}(k) > 0$. Remember that $p_t$ is the optimal price at time $t$. If $p \leq p_t$, 
\begin{align*}
\frac{\sum_{k \geq p} w^b_{j,t+1}(k)}{\sum_{k < p} w^b_{j,t+1}(k)}  
= 
\frac{\sum_{k = p}^{p_t-1} w^b_{j,t}(k) +  \sum_{k \geq p_t} w^b_{j,t}(k) \exp(\eta q^b_t (\vb_j - p_t))}
{\sum_{k < p} w^b_{j,t}(k)}
 \geq \frac{\sum_{k \geq p} w^b_{j,t}(k)}{\sum_{k < p} w^b_{j,t}(k)}.
\end{align*}
 When $p > p_t$, 
 \begin{align*}
\frac{\sum_{k \geq p} w^b_{j,t+1}(k)}{\sum_{k < p} w^b_{j,t+1}(k)}  
&= 
\frac{\sum_{k \geq p} w^b_{j,t}(k) \exp(\eta q^b_t (\vb_j - p_t))}{\sum_{k < p_t} w^b_{j,t}(k) + \sum_{k = p_t}^{p-1} w^b_{j,t}(k) \exp(\eta q^b_t (\vb_j - p_t))}
\\&\geq 
\frac{\sum_{k \geq p} w^b_{j,t}(k) \exp(\eta q^b_t (\vb_j - p_t))}{\left(\sum_{k < p_t} w^b_{j,t}(k) + \sum_{k = p_t}^{p-1} w^b_{j,t}(k) \right) \exp(\eta q^b_t (\vb_j - p_t))}
\\& = \frac{\sum_{k \geq p} w^b_{j,t}(k)}{\sum_{k < p} w^b_{j,t}(k)}. 
\end{align*}
Since 
\[
\sum_{k \geq p} w^b_{j,t+1}(k) + \sum_{k < p} w^b_{j,t+1}(k) = 1, \sum_{k \geq p} w^b_{j,t}(k) + \sum_{k < p} w^b_{j,t}(k) = 1,
\]
we have that for all $p$,
 \begin{align*}
\frac{\sum_{k \geq p} w^b_{j,t+1}(k)}{1 - \sum_{k \geq p} w^b_{j,t+1}(k)}  
&\geq  \frac{\sum_{k \geq p} w^b_{j,t}(k)}{1-\sum_{k \geq p} w^b_{j,t}(k)}.
\end{align*}
This in particular implies that for all $p$,
\[
\sum_{k \geq p} w^b_{j,t+1}(k) \left(1 - \sum_{k \geq p} w^b_{j,t}(k)\right) \geq \sum_{k \geq p} w^b_{j,t}(k) \left(1 - \sum_{k \geq p} w^b_{j,t+1}(k)\right),
\]
hence
\[
\sum_{k \geq p} w^b_{j,t+1}(k) \geq \sum_{k \geq p} w^b_{j,t}(k).
\]
 \end{proof}

\begin{lemma}[Update moves mass up by a constant amount] \label{lem:mass} Suppose at time $t$, at least one buyer and one seller can trade.
There exists a constant $C(\varepsilon) > 1$ such that for any buyer $j$ with $\vb_j > p_t$  and $\sum_{k = p_t}^{\vb_j} w^b_{j,t}(k) \leq 1-\varepsilon$, we have that 
\[
\frac{\sum_{k = p_t}^{\vb_j} w^b_{j,t+1}(k)}{\sum_{k = p_t}^{\vb_j} w^b_{j,t}(k)} \geq C(\varepsilon).
\]
\end{lemma}
\begin{proof}
Let  $X_t(p)$ be the probability that buyer $j$ bids at least $p$ on round $t$. For simplicity of notations, we omit the $j$ subscripts in the proof.
Trivially:
	\begin{align*}
	X_t(p_t) = \sum_{k=p_t}^{\vb_j} w^b_{j,t}(k).
	\end{align*}
	Now, by the definition of exponential weights, we have that 
	\begin{align*}
	X^{t+1}(p_t) = \frac{e^{\eta q^b_t (\vb_j-p_t)} X_t(p_t)}{e^{\eta q^b_t (\vb_j-p_t)} X_t(p_t) + (1-X_t(p_t))}
	\end{align*}
since the buyer updates $w^b_{j,t}(k)$ with $e^{\eta q^b_t (\vb_j-p_t)}$ for all bids $k$ above $p_t$ up to $\vb_j$, and updates weights on bids $k \leq  p_t$  with $e^{\eta q^b_t \cdot 0} = 1$. It immediately follows that
	\begin{align*}
	\frac{X^{t+1}(p_t)}{X_t(p_t)} = \frac{e^{\eta q^b_t(\vb_j-p_t)}}{X_t(p_t) (e^{\eta q^b_t (\vb_j-p_t)}-1) + 1}
	\end{align*}
Now by assumption, $X_t(p_t)  < 1-\epsilon$, so 
	\begin{align*}
	\frac{X^{t+1}(p_t)}{X_t(p_t)} &> \frac{e^{\eta q^b_t(\vb_j-p_t)}}{(1-\epsilon) (e^{\eta q^b_t (\vb_j-p_t)}-1) + 1} 
	\\&=\frac{e^{\eta q^b_t (\vb_j-p_t)}}{e^{\eta q^b_t(\vb_j-p_t)} -\epsilon e^{\eta q^b_t(\vb_j-p_t)} + \epsilon } 
	\\& = \frac{e^{\eta q^b_t(\vb_j-p_t)}}{e^{\eta q^b_t(\vb_j-p_t)} + \epsilon(1-e^{\eta q^b_t(\vb_j-p_t)})} \\
	&= \frac{1}{1-\epsilon(1-\frac{1}{e^{\eta q^b_t(\vb_j-p_t)}})}
	\end{align*}
Using the fact that $q^b_t \geq \frac{1}{n^b}$, as there are at most $n^b$ buyers and at least one possible seller to trade with, and the fact that $\vb_j-p_t \geq 1$, we get that
\begin{align*}
	e^{\eta/n^b} \leq e^{\eta q^b_t(\vb_j-p_t)}.
\end{align*} 	
In turn, 
\begin{align*}
	\frac{X^{t+1}(p_t)}{X_t(p_t)} 
	\geq \frac{1}{1-\epsilon(1-e^{-\eta/n^b})} > 1.
\end{align*}
Letting $C(\varepsilon) = \frac{1}{1-\epsilon(1-e^{-\eta/n^b})}$ is enough to conclude the proof.
\end{proof}

\begin{cor}\label{cor:limto1}
Pick any buyer $j$, and let $p < \vb_j$. Let $N_t(p)$ be the number of times price $p$ is picked by the mechanism so that at least one trade is possible at $p$, up until time $t$. In other words,
$$
N_t (p) = \sum_{t' \le t} \1 \left[ \Pi_{t'} \left(p, \ts_{t'}, \tb_{t'} \right) \ge 1 \right]
$$
If $\lim_{t \to \infty} N_t(p) = +\infty$, then
\begin{align*}
	\lim_{t \to \infty} \Pr[\tb_{j,t} \geq p] = 1
\end{align*}
\end{cor}
\begin{proof}

	Fix $\epsilon>0$. At time $t$, by applying Lemma \ref{lem:mass} and Lemma~\ref{lem: non-decreasing} repeatedly, we have that  
	\begin{align*}
	\Pr[\tb_{j,t} \geq  p] &\geq \min\{1-\epsilon, C(\epsilon)^{N_t(p)} \Pr[\tb_{j,0} \geq p ]\} 
	\\&\geq \min \{1-\epsilon, C(\epsilon)^{N_t(p)} \frac{1}{V}\}  
	\end{align*}
	where the last inequality follows because the initial weights are uniform over all bids. By assumption, there exists $T$ such that for all $t \ge T$: $N_t(p) \geq \frac{\log \left((1-\epsilon)(V)\right)}{\log C(\epsilon)}$, and consequently,
	\[
	\Pr[\tb_{j,t} \geq  p] \geq 1-\varepsilon.
	\]
	Since this holds for every $\epsilon>0$, the limit statement follows. 
\end{proof}

We note that a similar Corollary exists for sellers as well. Now, we need to show that there is a price that clears benchmark $\opt'$ and is chosen by the mechanism infinitely often. This is the object of Lemma~\ref{lem:goodevent}, whose proof relies on Claim~\ref{clm: bound_weights}. Once again, we state the Claim only for buyers and note that a similar result for sellers as well.

\begin{clm}\label{clm: bound_weights}
For every buyer $j$, for all $t$, $\frac{1}{V} \leq w^b_{j,t}(\vb_j)  \leq \frac{1}{2}$.
\end{clm}

\begin{proof}
At time step $t$, if $p_t > \vb_j$, agent $j$ does not update any weight. If $p_t \leq \vb_j$, it is easy to see that the weight on $\vb_j$ cannot decrease in the next round. Indeed, for any $k$ such that $p_t \leq k \leq \vb_j$, we have that
\begin{align*}
w^b_{j,t+1}(k) 
&= w^b_{j,t}(k) \cdot \frac{\exp\left(\eta q^b_t (\vb_j - p_t)\right)}{\sum_{k < p_t} w^b_{j,t}(k) + \sum_{k \geq p_t} w^b_{j,t}(k) \exp\left(\eta q^b_t (\vb_j - p_t)\right)}
\\&= w^b_{j,t}(k) \cdot \frac{1}{\exp\left(-\eta q^b_t (\vb_j - p_t)\right) \sum_{k < p_t} w^b_{j,t}(k) + \sum_{k \geq p_t} w^b_{j,t}(k)}
\\& = w^b_{j,t}(k) \cdot \frac{1}{\exp\left(-\eta q^b_t (\vb_j - p_t)\right) \sum_{k < p_t} w^b_{j,t}(k) + 1 - \sum_{k < p_t} w^b_{j,t}(k)}
\\& = w^b_{j,t}(k) \cdot \frac{1}{1 - \left(1-\exp\left(-\eta q^b_t (\vb_j - p_t)\right)\right) \sum_{k < p_t} w^b_{j,t}(k)}
\\&\geq w^b_{j,t}(k),
\end{align*}
where the last step follows from noting that both $1-\exp\left(-\eta q^b_t (\vb_j - p_t)\right),~\sum_{k < p_t} w^b_{j,t}(k) \leq 1$. As such, $w^b_{j,t}(\vb_j)$ is non-decreasing in $t$, so $w^b_{j,t}(\vb_j) \geq w^b_{j,0}(\vb_j) = \frac{1}{V}$. 

Let us now prove the second inequality. Note that at any time step $t$, let $p_t$ be the price chosen by the mechanism. When $p_t > \vb_j$, $j$ does not update his weight. Similarly, when $p_t = \vb_j$, the exponential update rule is the same for $w^b_{j,t}(\vb_j)$ and $w^b_{j,t}(\vb_j - 1)$ and given by $\exp \left(\eta q^b_t (\vb_j - p_t) \right) = \exp \left(0 \right) = 1$. When $p_t < \vb_j$, both $w^b_{j,t}(\vb_j)$ and $w^b_{j,t}(\vb_j - 1)$ are multiplied by the same amount $\exp \left(\eta q^b_t (\vb_j - p_t) \right)$. Therefore, it immediately follows by induction that $w^b_{j,t}(\vb_j) = w^b_{j,t}(\vb_j - 1)$ for all $t$. In particular, this implies $w^b_{j,t}(\vb_j) \leq 1/2$, as  $w^b_{j,t}(\vb_j) + w^b_{j,t}(\vb_j - 1) \leq 1$.
\end{proof}

\begin{lemma}[Good event] \label{lem:goodevent}
Suppose $\opt' > 0$, and let 
 \[
\gamma \triangleq \left(\frac{1}{V}\right)^{1 + |n^b(\nu^s+1)| |n^s(\nu^b-1)|} \cdot \left(\frac{1}{2}\right)^{\left(n^b - |n^b(\nu^s+1)|\right)\left(n^s - |n^s(\nu^b-1)|\right)} > 0.
 \]
At any time $t$, $\nu^s < p_t < \nu^b$ and at least one trade is possible with probability at least $\gamma$.
\end{lemma}

\begin{proof}
By Claim~\ref{clm: bound_weights}, we have that with probability at least 
\[
\left(\frac{1}{V}\right)^{|n^b(\nu^s+1)| |n^s(\nu^b-1)|} \cdot \left(\frac{1}{2}\right)^{\left(n^b - |n^b(\nu^s+1)|\right)\left(n^s - |n^s(\nu^b-1)|\right)} = V \gamma,
\]
all buyers with value $\vb_j > \nu^s$ bid their value, all buyers with value $\vb_j \leq \nu^s$ bid strictly below their value, all sellers with value $\vs_i < \nu^b$ bid their value, and all sellers with $\vs_i \geq \nu^b$ bid strictly more than their value. In particular, since $\nu^s < \nu^b$, all buyers with value $\vb_j \geq \nu^b$ bid their value and all sellers with value $\vs_i \leq \nu^s$ bid their value. By definition of $\nu^b$ and $\nu^s$, there are at least $\opt'$ such buyers and sellers, so setting any price $p$ satisfying $\nu^s \leq p \leq \nu^b$ clears $\opt'$ shares at least. On the other hand, any price $p > \nu^b$ and any price $p < \nu^s$ cannot clear $\opt'$ shares. Therefore, $\nu^s \leq p_t \leq \nu^b$. Further, since all buyers with value $\vb_j \geq \nu^b$ and all sellers with value $\vs_i \leq \nu^s$ bid their values, and $\nu^b \geq \nu^s$, at least $\opt' \geq 1$ trades happen at price $p_t$. 

When $\nu^s < p < \nu^b$ for all optimal prices, this is enough to conclude the proof. Now, suppose $p = \nu^b$ is an optimal price at time $t$. By construction, no seller bids $\nu^b$. As such, the number of sellers with bids under $p$ and the number of sellers with bids under $p-1$ are the same, and $p- 1 = \nu^b - 1$ clears at least as many shares as $p$, hence is optimal at time $t$. Because $p_t$ is chosen uniformely at random among the set of optimal prices, and there are at most $V$ optimal prices, $p-1$ is picked with probability at least $\frac{1}{V}$, and satisfies $\nu^s < p-1 < \nu^b$ by Claim~\ref{clm: buyer-seller gap}. Similarly, if $p = \nu^s$ is optimal, then so is $p + 1 < \nu^b$, and it is picked by the mechanism with probability at least $\frac{1}{V}$. This concludes the proof.
\end{proof}

We are now ready to put everything together, and show Theorem \ref{thm:limit_1}. 
\begin{proof}[Proof of Theorem \ref{thm:limit_1}]
The case when $\opt' = 0$ is immediate. So let us assume $\opt' > 0$. Lemma \ref{lem:goodevent} shows that at any given round, there is a constant probability $\gamma > 0$ to pick $p_t \in (\nu^s,\nu^b)$ and realize at least one trade at that price. As such, as $t \to +\infty$, the number of times the mechanism picks a price in $(\nu^s,\nu^b)$ such that a trade is realized also tends to infinity. In particular, by the pigeonhole principle, there exists a price $p^\star \in (\nu^s,\nu^b)$ such that
\[
\lim_{t \to \infty} N_t(p^\star) = +\infty.
\]
By Corollary~\ref{cor:limto1}, for every buyer $j \in n^b(\nu^b)$,
\[
\lim_{t \to \infty} \Pr[\tb_{j,t} \geq p^\star] = 1,
\]
and similarly, for every seller $i \in n^s(\nu^s)$,
\[
\lim_{t \to \infty} \Pr[\ts_{i,t} \leq p^\star] = 1.
\]
Since there are at least $\opt'$ buyers in $n^b(\nu^b)$ and $\opt'$ sellers in $n^s(\nu^s)$, we have that 
\[
1 \ge \Pr \left[\Pi_t\left(p_t,\ts,\tb\right) \geq \opt' \right] \geq \prod_{j \in n^b(\nu^b)}  \Pr[\tb_{j,t} \geq p^\star] \cdot \prod_{i \in n^s(\nu^s)}  \Pr[\ts_{i,t} \leq p^\star],
\]
which concludes the proof.
\end{proof}

\subsection{Proof of Theorem~\ref{thm:limit}}\label{app:dynamics_notindifferent}
The proof is similar to that of Theorem~\ref{thm:limit_1}, and is given below. We start by showing in Corollary~\ref{cor:limto1_variant} that if a price $p$ is picked by the mechanism infinitely many times, every buyer with value at least $p$ learns to bid higher than $p$ with probability going to $1$. 

\begin{lemma}\label{lem: non-decreasing_variant}
For all $t$, for all $p \in [V]$, for all buyers $j$,
\[
\sum_{k = p}^{\vb_j} \wb_{j,t+1}(k) \geq \sum_{k = p}^{\vb_j} \wb_{j,t}(k).
\]
\end{lemma}

\begin{proof}
The proof is identical to that of Lemma~\ref{lem: non-decreasing}.
 \end{proof}

We then characterize by how much the weight allocated to bids above the chosen price $p_t$ increase for a buyer $j$, at every time step $t$:
\begin{lemma}[Update moves mass up by a constant amount] \label{lem:mass_variant}  Suppose at time $t$, at least one buyer and one seller can trade. There exists a constant $C(\varepsilon) > 1$ such that for any buyer $j$ with $\vb_j \geq p_t$  and $\sum_{k = p_t}^{\vb_j} \wb_{j,t}(k) \leq 1-\varepsilon$, we have that 
\[
\frac{\sum_{k = p_t}^{\vb_j} \wb_{j,t+1}(k)}{\sum_{k = p_t}^{\vb_j} \wb_{j,t}(k)} \geq C(\varepsilon).
\]
\end{lemma}
\begin{proof}
Note that when $p_t < \vb_j$, we have by Lemma~\ref{lem:mass} that
\begin{align*}
\frac{\sum_{k = p_t}^{\vb_j} \wb_{j,t+1}(k)}{\sum_{k = p_t}^{\vb_j} \wb_{j,t}(k)}
&\geq \frac{1}{1-\epsilon(1-e^{-\eta/n^b})} > 1.
\end{align*}
Now, when $p_t = \vb_j$ note that 
\begin{align*}
\frac{\sum_{k = p_t}^{\vb_j} \wb_{j,t+1}(k)}{\sum_{k = p_t}^{\vb_j} \wb_{j,t}(k)} 
= \frac{\wb_{j,t+1}(\vb_j)}{\wb_{j,t}(\vb_j)}
= \exp\left(\eta q_t^b \xi\right).
\end{align*}
In particular, as there is at least one possible trade, we have that $q_t^b \geq 1/n^b$, hence 
\[
\frac{\sum_{k = p_t}^{\vb_j} \wb_{j,t+1}(k)}{\sum_{k = p_t}^{\vb_j} \wb_{j,t}(k)} \geq \exp\left(\frac{\eta \xi}{n^b}\right).
\]
Letting $C(\varepsilon) = \min\left(\frac{1}{1-\epsilon(1-e^{-\eta/n^b})},\exp\left(\frac{\eta \xi}{n^b}\right) \right)$ is enough to conclude the proof.
\end{proof}

\begin{cor}\label{cor:limto1_variant}
Pick any buyer $j$, and let $p \leq \vb_j$. Let $N_t(p)$ be the number of times price $p$ is picked and at least one trade is possible at price $p$, up until time $t$. If $\lim_{t \to \infty} N_t(p) = +\infty$, then
\begin{align*}
	\lim_{t \to \infty} \Pr[\tb_{j,t} \geq p] = 1
\end{align*}
\end{cor}
\begin{proof}
This is identical to the proof of Corollary~\ref{cor:limto1}.
\end{proof}

Second, we need to show that there is a price that clears benchmark $\opt'$ and is chosen by the mechanism infinitely often. 

\begin{lemma}[Good event] \label{lem:goodevent_variant}
With probability at least $\left(\frac{1}{V}\right)^{n^b + n^s}$, all agents bid their valuation.
\end{lemma}

\begin{proof}
By the same proof as Corollary~\ref{clm: bound_weights}, for every agent $j$ and for all $t$, $\frac{1}{V} \leq \wb_{j,t}(\vb_j)$. This is enough to prove the lemma.
\end{proof}

We are now ready to put everything together, and show Theorem \ref{thm:limit}. 
\begin{proof}[Proof of Theorem \ref{thm:limit}]
Suppose $\opt > 0$ (otherwise the result is immediate). When all agents bid their values, the mechanism selects a price that executes $\opt \geq 1$ trades. Lemma \ref{lem:goodevent_variant} shows this happens with constant probability at any given round, and as such happens infinitely often when the number of rounds goes to infinity. By the pigeonhole principle, there exists a price $p^\star$ such that there are at least $\opt$ buyers (resp. sellers) with value at least (resp. at most) $p^\star$, and such that 
\[
\lim_{t \to \infty} N_t(p^\star) = +\infty.
\]
By Corollary~\ref{cor:limto1_variant}, for any buyer with $\vb_j \geq p^\star$,
\[
\lim_{t \to \infty} \Pr[\tb_{j,t} \geq p^\star] = 1,
\]
and similarly, for every seller $i$ with $\vs_i \leq p^\star$,
\[
\lim_{t \to \infty} \Pr[\ts_{i,t} \leq p^\star] = 1.
\]
In turn, since
\[
1 \ge \Pr \left[\Pi_t \left(p_t,\ts_t,\tb_t\right) \geq \opt \right] 
\geq \prod_{j \in [n^b]:~\vb_j \geq p^\star}  \Pr[\tb_{j,t} \geq p^\star] \cdot \prod_{i \in [n^s]:~\vs_i \leq p^\star}  \Pr[\ts_{i,t} \leq p^\star],
\]
we have
\[
\lim_{t \to +\infty} \Pr \left[\Pi_t \left(p_t,\ts_t,\tb_t\right) \geq \opt \right] = 1.
\]
\end{proof}

\subsection{Proof of Theorem~\ref{thm:limit_private}}\label{app:dynamics_private}
We start by noting that in the private case, the weights are still non-decreasing over time.
\begin{lemma}\label{lem: non-decreasing_variant_private}
For all $t$, for all $p \in [V]$, for all buyers $j$,
\[
\sum_{k = p}^{\vb_j} \wb_{j,t+1}(k) \geq \sum_{k = p}^{\vb_j} \wb_{j,t}(k).
\]
\end{lemma}

One distinction compared to the non-private case arises with respect to the amount by which the weights above $p_t$ are updated. This amount depends on $q^b_t$, which is a random variable over the randomness of private computation of the selection probability. We note that \emph{conditionally on $q^b_t \geq 1/\nb$}, Lemma~\ref{lem:mass} carries through, as formalized below: 
\begin{lemma}[Update moves mass up by a constant amount] \label{lem:mass_variant_private}
Suppose at time $t$, at least one buyer and one seller can trade and that $q_t^b > 1/n^b$. There exists a constant $C(\varepsilon) > 1$ such that for any buyer $j$ with $\vb_j \geq p_t$  and $\sum_{k = p_t}^{\vb_j} \wb_{j,t}(k) \leq 1-\varepsilon$,
\[
\frac{\sum_{k = p_t}^{\vb_j} \wb_{j,t+1}(k)}{\sum_{k = p_t}^{\vb_j} \wb_{j,t}(k)} \geq C(\varepsilon).
\]
\end{lemma}

We now fix a price $p$. We show that for one such $p$, if the event where $p$ is the price picked by the mechanism \emph{and} $q^b_t \geq 1/n$ happens infinitely often, then all bidders with valuation equal to or larger than $p$ learn to bid higher than $p$ with probability that goes to $1$.
\begin{lemma}\label{cor:limto1_variant_private}
Pick any buyer $j$, and let $p \leq \vb_j$. Let $N_t(p)$ be the number of times price $p$ is picked by the mechanism so that at least one trade is possible at $p$ and $q^b > 1/n^b$, up until time $t$. In other words,
$$
N_t (p) = \sum_{t' \le t} \1 \left[ \Pi_{t'} \left(p, \ts_{t'}, \tb_{t'} \right) \ge 1 , q_{t'}^b > \frac{1}{n^b} \right]
$$
If $\lim_{t \to \infty} N_t(p) = +\infty$, then $\lim_{t \to \infty} \Pr[\tb_{j,t} \geq p] = 1$.
\end{lemma}

We note that the event in which all agents bid their valuation and the mechanism (despite the randomness due to privacy) picks an optimal price $p$ and releases $q^b \geq 1/\nb$, $q^s \geq 1/\ns$ happens with at least constant probability (independent of the time dimension of the problem), hence infinitely many times when the time horizon goes to infinity:

\begin{lemma}[Good event] \label{lem:goodevent_variant_private}
Suppose $\opt \ge 1$. At any round $t$, with probability at least $C\left(\frac{1}{V}\right)^{n^b + n^s + 1}$ for some constant $C > 0$: all buyers bid their valuations, $q_t^b > 1/n^b$, $q_t^s > 1/n^s$, and the chosen price $p_t$ is an optimal price that clears $OPT$ shares.
\end{lemma}

\begin{proof}
We have shown before in the proof of Lemma \ref{lem:goodevent_variant} that with probability at least $V^{-(n^b+n^s)}$ every agent bids their valuation. 

In the rest of the proof, we condition on all agents bidding their valuations in the current round $t$. Conditional on this, we show that with constant probability, simultaneously: $q_t^b > 1/n^b$, and $q_t^s > 1/n^s$. Recall from Algorithm \ref{alg:mech} that in each round $t$, given the selected price $p_t$, we have 
$$
q_t^b = \min \left(1,\frac{\left( \widehat{s}_t \right)_+ }{\left( \widehat{b}_t - \frac{\ln (1/\alpha)}{\epsilon} \right)_+ }\right),
~q_t^s = \min\left(1,\frac{\left( \widehat{b}_t \right)_+ }{\left( \widehat{s}_t - \frac{\ln (1/\alpha)}{\epsilon} \right)_+ }\right).
$$
By the accuracy guarantees of the Laplace mechanism and the fact that Laplace noise has positive value with probability $1/2$, we have that with constant probability $C$ (for some $C$ that only depends on $\alpha$ and $\epsilon$ but not on $t$), the 4 following events simultaneously hold:
\begin{enumerate}
\item 
$\left( \widehat{b}_t - \frac{\ln (1/\alpha)}{\epsilon} \right)_+ \leq \sum_{j \in \B} \mathbbm{1}\left[\vb_j \geq p_t\right] \leq \nb$,
\item $\widehat{b}_t \geq \sum_{j \in B} \mathbbm{1}\left[\vb_j \leq p_t\right] \ge \opt \ge 1$,
\item
$\left( \widehat{s}_t - \frac{\ln (1/\alpha)}{\epsilon} \right)_+ \leq \sum_{i \in \S} \mathbbm{1}\left[\vs_i \leq p_t\right] \leq \ns$,
\item $\widehat{s}_t \geq \sum_{i \in \S} \mathbbm{1}\left[\vs_i \leq p_t\right] \ge \opt \ge 1$, noting that at least one trade is possible at price $p_t$.
\end{enumerate}
Using the above inequalities, we obtain that with probability $C$, 
\begin{align*}
&q^b_t = \min \left(1,\frac{\left( \widehat{s}_t \right)_+ }{\left( \widehat{b}_t - \frac{\ln (1/\alpha)}{\epsilon} \right)_+ }\right) \geq \min\left(1,\frac{1}{\nb}\right) \geq \frac{1}{\nb},\\
&q^s_t= \min\left(1,\frac{\left( \widehat{b}_t \right)_+ }{\left( \widehat{s}_t - \frac{\ln (1/\alpha)}{\epsilon} \right)_+ }\right) \geq \min\left(1,\frac{1}{\ns}\right)\geq \frac{1}{\ns}.
\end{align*}

To finish the proof, we just need to show that, conditional on all agents bidding their valuations in the current round $t$, with probability at least $1/V$, $p_t$ -- the price selected when every agent bids their valuation -- is an optimal price. Note there exists a price $p^\star_t$ that is optimal for round $t$, i.e. such that $\Pi_t(p^\star_t, \vs, \vb) \geq \Pi_t(p, \vs, \vb)$ for all $p$. By the exponential mechanism, this price $p^\star_t$ is selected with probability
\begin{align*}
\frac{\exp\left(\epsilon \Pi_t(p^\star_t, \vs, \vb) / 2\right)}{\sum_{p=1}^V \exp\left(\epsilon \Pi_t(p, \vs, \vb) / 2\right)}
 \geq \frac{\exp\left(\epsilon \Pi_t(p^\star_t, \vs, \vb) / 2\right)}{\sum_{p=1}^V \exp\left(\epsilon \Pi_t(p^\star_t, \vs, \vb) / 2\right)}
 = \frac{1}{V}.
\end{align*}
\end{proof}

We are now ready to put everything together, and show Theorem \ref{thm:limit_private}. 
\begin{proof}[Proof of Theorem \ref{thm:limit_private}]
Let us for simplicity call:
$$
f(\epsilon, \alpha) \triangleq \frac{2 \ln (V/\alpha)}{\epsilon} - \frac{2 \ln \left( 1/\alpha \right)}{\epsilon}  - \sqrt{6  \left(\opt + \frac{\ln (1/\alpha)}{\epsilon} \right) \ln \left(1/\alpha\right)}
$$
Suppose $\opt > 0$ (otherwise the result is immediate). By Lemma~\ref{lem:goodevent_variant_private} that shows that with constant probability (independent of time) in every round, the mechanism picks an optimal price and $q^b,q^s \geq \frac{1}{n}$, this event must happen infinitely many times. By the pigeonhole principle, there exists an optimal price $p^\star$ such that infinitely many times, $p^\star$ is picked by the mechanism with $q^b,q^s \geq \frac{1}{n}$. In turn, all buyers $j$ with $\vb_j \geq p^\star$ and all sellers $i$ with $\vs_i \leq p^\star$ (there are at least $OPT$ of them, since $p^\star$ is optimal) learn to bid above, respectively below price $p^\star$ with probability that tends to $1$ as $t$ goes to infinity, by Lemma~\ref{cor:limto1_variant_private}. Formally, for every buyer $j$ with $\vb_j \geq p^\star$,
\[
\lim_{t \to \infty} \Pr[\tb_{j,t} \geq p^\star] = 1,
\]
and similarly, for every seller $i$ with $\vs_i \leq p^\star$, we have that
\[
\lim_{t \to \infty}  \Pr[\ts_{i,t} \leq p^\star] = 1.
\]
Hence 
$$
\lim_{t \to \infty} \prod_{j \in [n^b]:~\vb_j \ge p^\star} \Pr[\tb_{j,t} \geq p^\star] \cdot \prod_{i \in [n^s]:~\vs_i \leq p^\star}  \Pr[\ts_{i,t} \leq p^\star] = 1
$$ 
and consequently, there exists $N(\alpha)$ large enough such that for all $t \geq N(\alpha)$,
\[
\prod_{j \in [n^b]:~\vb_j \ge p^\star} \Pr[\tb_{j,t} \geq p^\star] \cdot \prod_{i \in [n^s]:~\vs_i \leq p^\star}  \Pr[\ts_{i,t} \leq p^\star] \geq 1-\alpha.
\]

When all buyers with value at least the price and sellers with value at most the price bid between their valuation and $p^\star$, the optimal number of shares that can be cleared is OPT. By the accuracy guarantee of Mechanism~\ref{alg:mech}, it must then be the case that for all $t \geq N(\alpha)$,
\begin{align*}
\Pr \left[ \Pi_t(p_t,\ts_t,\tb_t) \geq \opt - f(\epsilon,\alpha) \right] 
&\geq (1-8\alpha) \prod_{j \in [n^b]:~\vb_j \ge p^\star} \Pr[\tb_{j,t} \geq p^\star] \cdot \prod_{i \in [n^s]:~\vs_i \leq p^\star}  \Pr[\ts_{i,t} \leq p^\star]
\\&\geq (1-8\alpha)(1-\alpha)
\\&\geq 1 - 9\alpha.
\end{align*}
This concludes the proof.
\end{proof}

The proof for benchmark $\opt'$ follows the same argument, and is omitted for simplicity of exposition. 

%% file: main.bbl
\begin{thebibliography}{}

\bibitem[Budish et~al., 2015]{budish}
Budish, E., Cramton, P., and Shim, J. (2015).
\newblock { The High-Frequency Trading Arms Race: Frequent Batch Auctions as a
  Market Design Response}.
\newblock {\em The Quarterly Journal of Economics}, 130(4):1547--1621.

\bibitem[Challet, 2019]{challet2019strategic}
Challet, D. (2019).
\newblock Strategic behaviour and indicative price diffusion in paris stock
  exchange auctions.
\newblock In {\em New Perspectives and Challenges in Econophysics and
  Sociophysics}, pages 3--12. Springer.

\bibitem[Challet and Gourianov, 2018]{challet2018dynamical}
Challet, D. and Gourianov, N. (2018).
\newblock Dynamical regularities of us equities opening and closing auctions.
\newblock {\em Market Microstructure and Liquidity}, 4(1).

\bibitem[Chen et~al., 2018]{chen}
Chen, Z., Ni, T., Zhong, H., Zhang, S., and Cui, J. (2018).
\newblock Differentially private double spectrum auction with approximate
  social welfare maximization.
\newblock {\em arXiv preprint arXiv:1810.07873}.

\bibitem[Cummings et~al., 2015]{cummings2015privacy}
Cummings, R., Kearns, M., Roth, A., and Wu, Z.~S. (2015).
\newblock Privacy and truthful equilibrium selection for aggregative games.
\newblock In {\em International Conference on Web and Internet Economics},
  pages 286--299. Springer.

\bibitem[Dong et~al., 2018]{dong2018strategic}
Dong, J., Roth, A., Schutzman, Z., Waggoner, B., and Wu, Z.~S. (2018).
\newblock Strategic classification from revealed preferences.
\newblock In {\em Proceedings of the 2018 ACM Conference on Economics and
  Computation}, pages 55--70.

\bibitem[Dwork et~al., 2006]{dwork}
Dwork, C., McSherry, F., Nissim, K., and Smith, A. (2006).
\newblock Calibrating noise to sensitivity in private data analysis.
\newblock In Halevi, S. and Rabin, T., editors, {\em Theory of Cryptography},
  pages 265--284, Berlin, Heidelberg. Springer Berlin Heidelberg.

\bibitem[Dwork and Roth, 2014]{aaron}
Dwork, C. and Roth, A. (2014).
\newblock The algorithmic foundations of differential privacy.
\newblock {\em Foundations and Trends® in Theoretical Computer Science},
  9(3–4):211--407.

\bibitem[Dwork et~al., 2010]{DworkRV10}
Dwork, C., Rothblum, G.~N., and Vadhan, S. (2010).
\newblock Boosting and differential privacy.
\newblock In {\em Proceedings of the 2010 IEEE 51st Annual Symposium on
  Foundations of Computer Science}, FOCS '10, pages 51--60, Washington, DC,
  USA. IEEE Computer Society.

\bibitem[Even-Dar et~al., 2006]{michael}
Even-Dar, E., Kakade, S.~M., Kearns, M., and Mansour, Y. (2006).
\newblock (in)stability properties of limit order dynamics.
\newblock In {\em Proceedings of the 7th ACM Conference on Electronic
  Commerce}, EC ’06, page 120–129, New York, NY, USA. Association for
  Computing Machinery.

\bibitem[Gatheral, 2010a]{gatheral2010no}
Gatheral, J. (2010a).
\newblock No-dynamic-arbitrage and market impact.
\newblock {\em Quantitative finance}, 10(7):749--759.

\bibitem[Gatheral, 2010b]{gatheralslides}
Gatheral, J. (2010b).
\newblock Three models of market impact.
\newblock In {\em Market Microstructure and High Frequency Data}.

\bibitem[Hsu et~al., 2014]{billboard}
Hsu, J., Huang, Z., Roth, A., Roughgarden, T., and Wu, Z.~S. (2014).
\newblock Private matchings and allocations.
\newblock In {\em Proceedings of the Forty-Sixth Annual ACM Symposium on Theory
  of Computing}, STOC ’14, page 21–30, New York, NY, USA. Association for
  Computing Machinery.

\bibitem[Hsu et~al., 2016]{hsu2016jointly}
Hsu, J., Huang, Z., Roth, A., and Wu, Z.~S. (2016).
\newblock Jointly private convex programming.
\newblock In {\em Proceedings of the twenty-seventh annual ACM-SIAM symposium
  on Discrete algorithms}, pages 580--599. SIAM.

\bibitem[Kannan et~al., 2014]{stableDP}
Kannan, S., Morgenstern, J., Roth, A., and Wu, Z.~S. (2014).
\newblock Approximately stable, school optimal, and student-truthful
  many-to-one matchings (via differential privacy).
\newblock In {\em Proceedings of the twenty-sixth annual ACM-SIAM symposium on
  Discrete algorithms}, pages 1890--1903. SIAM.

\bibitem[Kearns et~al., 2014]{jointdp}
Kearns, M., Pai, M.~M., Roth, A., and Ullman, J. (2014).
\newblock Mechanism design in large games: Incentives and privacy.
\newblock {\em The American Economic Review}, 104(5):431--435.

\bibitem[Lewis, 2014]{lewis2014flash}
Lewis, M. (2014).
\newblock {\em Flash Boys: A Wall Street Revolt}.
\newblock A Wall Street Revolt. W. W. Norton.

\bibitem[McSherry and Talwar, 2007]{McSherryT07}
McSherry, F. and Talwar, K. (2007).
\newblock Mechanism design via differential privacy.
\newblock In {\em Proceedings of the 48th Annual IEEE Symposium on Foundations
  of Computer Science}, FOCS '07, pages 94--103, Washington, DC, USA. IEEE
  Computer Society.

\bibitem[NASDAQ, 2020]{nasdaq}
NASDAQ (2020).
\newblock Nasdaq opening and closing crosses.

\bibitem[NYSE, 2020]{nyse}
NYSE (2020).
\newblock Nyse opening and closing auctions fact sheet.

\bibitem[Parsons et~al., 2006]{parsons2006everything}
Parsons, S., Marcinkiewicz, M., Niu, J., and Phelps, S. (2006).
\newblock Everything you wanted to know about double auctions, but were afraid
  to (bid or) ask.

\bibitem[Parsons et~al., 2011]{parsons2011auctions}
Parsons, S., Rodriguez-Aguilar, J.~A., and Klein, M. (2011).
\newblock Auctions and bidding: A guide for computer scientists.
\newblock {\em ACM Computing Surveys (CSUR)}, 43(2):1--59.

\bibitem[Rogers et~al., 2015]{rogers2015inducing}
Rogers, R., Roth, A., Ullman, J., and Wu, Z.~S. (2015).
\newblock Inducing approximately optimal flow using truthful mediators.
\newblock In {\em Proceedings of the Sixteenth ACM Conference on Economics and
  Computation}, pages 471--488.

\bibitem[Rogers and Roth, 2014]{RR14}
Rogers, R.~M. and Roth, A. (2014).
\newblock Asymptotically truthful equilibrium selection in large congestion
  games.
\newblock In {\em Proceedings of the Fifteenth ACM conference on Economics and
  Computation}, pages 771--782.

\bibitem[Wah and Wellman, 2013]{WahWellman13}
Wah, E. and Wellman, M. (2013).
\newblock Latency arbitrage, market fragmentation, and efficiency: A two-market
  model.
\newblock In {\em Proceedings of the Thirteenth ACM conference on Economics and
  Computation}.

\end{thebibliography}
